\documentclass[pra,aps]{revtex4}
\usepackage{amsmath,amsfonts,amssymb,caption,hyperref,color,epsfig,graphics,graphicx,latexsym,mathrsfs,revsymb,theorem,url,verbatim,epstopdf,subcaption}
\allowdisplaybreaks[3]

\hypersetup{colorlinks,linkcolor={blue},citecolor={blue},urlcolor={red}}

\usepackage{hyperref}

\makeatletter

\newcommand{\Rmnum}[1]{\expandafter\@slowromancap\romannumeral #1@}
\makeatother
\newtheorem{definition}{Definition}
\newtheorem{proposition}[definition]{Proposition}
\newtheorem{Lemma}[definition]{Lemma}

\newtheorem{Theorem}[definition]{Theorem}
\newtheorem{Corollary}[definition]{Corollary}
\newtheorem{conjecture}[definition]{Conjecture}

\newtheorem{remark}[definition]{Remark}

\newtheorem{example}{Example}
\newtheorem{question}[definition]{Question}

\def\squareforqed{\hbox{\rlap{$\sqcap$}$\sqcup$}}
\def\qed{\ifmmode\squareforqed\else{\unskip\nobreak\hfil
		\penalty50\hskip1em\null\nobreak\hfil\squareforqed
		\parfillskip=0pt\finalhyphendemerits=0\endgraf}\fi}
\def\endenv{\ifmmode\;\else{\unskip\nobreak\hfil
		\penalty50\hskip1em\null\nobreak\hfil\;
		\parfillskip=0pt\finalhyphendemerits=0\endgraf}\fi}
\newenvironment{proof}{\noindent \textbf{{Proof.~} }}{\qed}
\def\Dbar{\leavevmode\lower.6ex\hbox to 0pt
	{\hskip-.23ex\accent"16\hss}D}
\makeatletter
\def\url@leostyle{%
	\@ifundefined{selectfont}{\def\UrlFont{\sf}}{\def\UrlFont{\small\ttfamily}}}
\makeatother
\def\bcj{\begin{conjecture}}
	\def\ecj{\end{conjecture}}
\def\bcr{\begin{corollary}}
	\def\ecr{\end{corollary}}
\def\bd{\begin{definition}}
	\def\ed{\end{definition}}
\def\bea{\begin{eqnarray}}
	\def\eea{\end{eqnarray}}
\def\bem{\begin{enumerate}}
	\def\eem{\end{enumerate}}
\def\bex{\begin{example}}
	\def\eex{\end{example}}
\def\bim{\begin{itemize}}
	\def\eim{\end{itemize}}
\def\bl{\begin{lemma}}
	\def\el{\end{lemma}}
\def\bma{\begin{bmatrix}}
	\def\ema{\end{bmatrix}}
\def\bpf{\begin{proof}}
	\def\epf{\end{proof}}
\def\bpp{\begin{proposition}}
	\def\epp{\end{proposition}}
\def\bqu{\begin{question}}
	\def\equ{\end{question}}
\def\br{\begin{remark}}
	\def\er{\end{remark}}
\def\bt{\begin{theorem}}
	\def\et{\end{theorem}}

\def\btb{\begin{tabular}}
	\def\etb{\end{tabular}}

\newcommand{\nc}{\newcommand}

\def\a{\alpha}

\def\d{\delta}
\def\e{\epsilon}

\def\m{\mu}

\def\r{\rho}
\def\s{\sigma}

\nc{\bbA}{\mathbb{A}} \nc{\bbB}{\mathbb{B}} \nc{\bbC}{\mathbb{C}}
\nc{\bbD}{\mathbb{D}} \nc{\bbE}{\mathbb{E}} \nc{\bbF}{\mathbb{F}}
\nc{\bbG}{\mathbb{G}} \nc{\bbH}{\mathbb{H}} \nc{\bbI}{\mathbb{I}}
\nc{\bbJ}{\mathbb{J}} \nc{\bbK}{\mathbb{K}} \nc{\bbL}{\mathbb{L}}
\nc{\bbM}{\mathbb{M}} \nc{\bbN}{\mathbb{N}} \nc{\bbO}{\mathbb{O}}
\nc{\bbP}{\mathbb{P}} \nc{\bbQ}{\mathbb{Q}} \nc{\bbR}{\mathbb{R}}
\nc{\bbS}{\mathbb{S}} \nc{\bbT}{\mathbb{T}} \nc{\bbU}{\mathbb{U}}
\nc{\bbV}{\mathbb{V}} \nc{\bbW}{\mathbb{W}} \nc{\bbX}{\mathbb{X}}
\nc{\bbZ}{\mathbb{Z}}

\nc{\bA}{{\bf A}} \nc{\bB}{{\bf B}} \nc{\bC}{{\bf C}}
\nc{\bD}{{\bf D}} \nc{\bE}{{\bf E}} \nc{\bF}{{\bf F}}
\nc{\bG}{{\bf G}} \nc{\bH}{{\bf H}} \nc{\bI}{{\bf I}}
\nc{\bJ}{{\bf J}} \nc{\bK}{{\bf K}} \nc{\bL}{{\bf L}}
\nc{\bM}{{\bf M}} \nc{\bN}{{\bf N}} \nc{\bO}{{\bf O}}
\nc{\bP}{{\bf P}} \nc{\bQ}{{\bf Q}} \nc{\bR}{{\bf R}}
\nc{\bS}{{\bf S}} \nc{\bT}{{\bf T}} \nc{\bU}{{\bf U}}
\nc{\bV}{{\bf V}} \nc{\bW}{{\bf W}} \nc{\bX}{{\bf X}}
\nc{\bZ}{{\bf Z}}

\nc{\cA}{{\cal A}} \nc{\cB}{{\cal B}} \nc{\cC}{{\cal C}}
\nc{\cD}{{\cal D}} \nc{\cE}{{\cal E}} \nc{\cF}{{\cal F}}
\nc{\cG}{{\cal G}} \nc{\cH}{{\cal H}} \nc{\cI}{{\cal I}}
\nc{\cJ}{{\cal J}} \nc{\cK}{{\cal K}} \nc{\cL}{{\cal L}}
\nc{\cM}{{\cal M}} \nc{\cN}{{\cal N}} \nc{\cO}{{\cal O}}
\nc{\cP}{{\cal P}} \nc{\cQ}{{\cal Q}} \nc{\cR}{{\cal R}}
\nc{\cS}{{\cal S}} \nc{\cT}{{\cal T}} \nc{\cU}{{\cal U}}
\nc{\cV}{{\cal V}} \nc{\cW}{{\cal W}} \nc{\cX}{{\cal X}}
\nc{\cZ}{{\cal Z}}

\nc{\hA}{{\hat{A}}} \nc{\hB}{{\hat{B}}} \nc{\hC}{{\hat{C}}}
\nc{\hD}{{\hat{D}}} \nc{\hE}{{\hat{E}}} \nc{\hF}{{\hat{F}}}
\nc{\hG}{{\hat{G}}} \nc{\hH}{{\hat{H}}} \nc{\hI}{{\hat{I}}}
\nc{\hJ}{{\hat{J}}} \nc{\hK}{{\hat{K}}} \nc{\hL}{{\hat{L}}}
\nc{\hM}{{\hat{M}}} \nc{\hN}{{\hat{N}}} \nc{\hO}{{\hat{O}}}
\nc{\hP}{{\hat{P}}} \nc{\hR}{{\hat{R}}} \nc{\hS}{{\hat{S}}}
\nc{\hT}{{\hat{T}}} \nc{\hU}{{\hat{U}}} \nc{\hV}{{\hat{V}}}
\nc{\hW}{{\hat{W}}} \nc{\hX}{{\hat{X}}} \nc{\hZ}{{\hat{Z}}}

\nc{\hn}{{\hat{n}}}

\def\max{\mathop{\rm max}}
\def\min{\mathop{\rm min}}

\def\rank{\mathop{\rm rank}}

\def\supp{\mathop{\rm supp}}

\newcommand{\bra}[1]{\langle#1|}
\newcommand{\ket}[1]{|#1\rangle}

\def\Dbar{\leavevmode\lower.6ex\hbox to 0pt
	{\hskip-.23ex\accent"16\hss}D}
\begin{document}
	\title{ Quantum Probabilistic Local Differential Privacy: Structural Properties and Sample Complexity Bounds}
	
	%\author{Xian Shi}\email[]
	%	{shixian01@gmail.com}
	%	\affiliation{College of Information Science and Technology,
		%	Beijing University of Chemical Technology, Beijing 100029, China}
	
	\author{Xian Shi}\email[]
{shixian01@gmail.com}
\affiliation{College of Information Science and Technology,
	Beijing University of Chemical Technology, Beijing 100029, China}

%\author{Yi Shen}\email[]
%{yishen@buaa.edu.cn}
%\affiliation{School of Mathematics and Systems Science, Beihang University, Beijing 100191, China}
%
%\author{Yize Sun}
%\affiliation{School of Mathematics and Systems Science, Beihang University, Beijing 100191, China}

%\author{Lijun Zhao}
%\affiliation{School of Mathematics and Systems Science, Beihang University, Beijing 100191, China}

%\author{Yumin Guo}
%\affiliation{School of Mathematical Sciences, Capital Normal University, Beijing 100048, China}

\date{\today}
\begin{abstract}
Differential privacy provides a rigorous framework for quantifying privacy leakage in data analysis, while its quantum extensions have become increasingly relevant with the development of quantum computing and quantum machine learning. In this work, we introduce and study quantum probabilistic local differential privacy, a relaxation of quantum local differential privacy in which the privacy constraint is allowed to fail on a spectral violation event with low probability.  This quantity can be interpreted as the probability under the quantum superoperation of a quantum privacy-loss violation, and is closely related to the acceptance probability of the quantum Neyman–Pearson test at a small threshold.

We investigate the basic structural properties of this privacy notion and clarify its relationship with existing forms of quantum differential privacy. We show the properties of quantum probabilistic local differential privacy under tensor-product composition and unitary post-processing, while it is in general neither convex nor closed under post-processing by arbitrary quantum channels. We further characterize when depolarizing noise satisfies quantum probabilistic local differential privacy under several representative scenarios. Finally, we connect quantum probabilistic privacy constraints with statistical inference by deriving a lower bound on probabilistically privatized contraction coefficients in terms of the hockey-stick divergence. As an application, we obtain sample complexity bounds of probabilistically privated asymmetric and symmetric quantum hypothesis testing. These results provide a systematic foundation for studying probabilistic privacy guarantees in quantum information processing and their operational consequences for private quantum statistical inference.

\end{abstract}

\maketitle

	\section{Introduction}
	The rapid growth of data and the remarkable progress of data-driven artificial intelligence have made privacy protection throughout the entire data lifecycle increasingly important. Differential privacy (DP) provides a mathematically rigorous and quantifiable framework for controlling privacy leakage when queries are performed on sensitive data  \cite{dwork2008differential,dwork2014}. Various extensions of DP have been proposed \cite{mironov2017renyi,cormode2018privacy,meiser2018approximate,wang2020comprehensive,balle2020hypothesis}. Among them, probabilistic differential privacy relaxes the standard privacy requirement by demanding that the prescribed privacy guarantee hold with high probability, while allowing it to fail on a suﬀiciently small set of outcomes.
	
With the rapid development of quantum computing and quantum machine learning, privacy protection in quantum information processing has become an increasingly important research topic. Inspired by classical differential privacy, quantum differential privacy has been introduced and investigated in  \cite{zhou2017differential,aaronson2019gentle,hirche2023quantum,guan2023detecting,angrisani2025quantum} inspired by the classical DP. Operationally, quantum differential privacy requires that the outputs of a quantum superoperator corresponding to neighboring input states be difficult to distinguish. Several variants of quantum differential privacy have also been proposed, including quantum local differential privacy (QLDP) \cite{hirche2023quantum} and quantum pufferfish privacy provide other variants of quantum differential privacy \cite{nuradha2024quantum}. QLDP requires privacy guarantees to hold uniformly over all admissible pairs of input quantum states. Applications of these quantum privacy notions to quantum machine learning have recently attracted considerable attention \cite{du2021quantum,watkins2023quantum,farokhi2023quantum,chen2024robust,rofougaran2024federated,nam2025quantum,yoshida2025mathematical,nuradha2026privacy}.  To the best of our knowledge, however, a systematic study of quantum probabilistic differential privacy has not yet been developed. In this work, we will consider quantum probabilistic local differential privacy (QPrLDP) from the perspective of quantum information. 
	
	Classical probabilistic differential privacy means that the probability of privacy loss exceeding a given threshold is at most $\delta$. 
	Motived by this construction, we characterize a superoperation $\mathcal{A}$ with $(\e,\delta)$-quantum probabilistic differential privacy as follows, $\mathrm{tr}\mathcal{A}(\rho)\{\mathcal{A}(\rho)-e^{\e}\mathcal{A}(\sigma)> 0\}\le \delta.$ This can be seen as the total probability of the corresponding spectral violation event on $\mathcal{A}(\rho)$. It also provides the acceptance probability on $\mathcal{A}(\rho)$ of the quantum Neyman-Pearson test at threshold $e^{\e}$. In quantum algorithm, one of the important problem is on the sample complexity of statistical inference under privacy constraints \cite{cheng2024sample,cheng2025invitation}. 
	Broadly speaking, sample complexity quantifies the minimum number of samples required to accomplish a prescribed task with a given accuracy or distinguishability guarantee. Previously, the authors in \cite{nuradha2025contraction} derived bounds on the sample complexity of quantum hypothesis testing under quantum local differential privacy. In contrast,  less is known about the corresponding sample complexity under QPrLDP.

	Our main contributions are summarized as follows.
	\begin{itemize}
		\item We introduce the quantum probabilistic local differential privacy and clarify its relationship with other existing notions of quantum differential privacy.
		\item  We establish several structural properties of QPrLDP. We study its behavior under tensor-product composition and post-processing by unitary operations. We also show that, in contrast to standard differential privacy, the set of QPrLDP channels is in general not convex and is not closed under post-processing by arbitrary quantum channels.
		\item We obtain the conditions when the depolarizing noise is QPrLDP under three scenarios.
	
		\item We connect QPrLDP with private quantum statistical inference. We derive a lower bound on the probabilistically privatized contraction coefficient in terms of the hockey-stick divergence. Using this result, we obtain sample complexity bounds for symmetric quantum hypothesis testing under QPrLDP.
	\end{itemize} 
	
	This work is organized as follows. In Section \ref{II}, we present the preliminary knowledge needed. In Section \ref{III}, we present the properties of QPrLDP under tensor product operations and post-processing with unitary operations. In Section \ref{IV}, we present the relations between QPrLDP and other variants of quantum local differential privacy. In Section \ref{V}, we pressent the conditions when the depolarizing noise is QPrLDP under three scenarios.
	In Section \ref{VI}, we obtain the sample complexity of $n$-outcome symmetric hypothesis testing with a bounded error under QPrLDP. In Section \ref{VII}, we ends the manuscript with a conclusion.

\section{Preliminary Knowledge}\label{II}
Assume $\mathcal{H}$ is a Hilbert space with finite dimensions, a unitary matrix $U$ on $\mathcal{H}$ is a matrix with $UU^{\dagger}=U^{\dagger}U=I$, here we denote $\mathbf{U}_{\mathcal{H}}$ as the set of all unitary matrices. A matrix $K$ on $\mathcal{H}$ is Hermitian if $K=K^{\dagger}.$ A matrix $M$ on $\mathcal{H}$ is positive semidefinite if 
 \begin{align*}
 	\bra{\phi}M\ket{\phi}\ge 0\hspace{4mm}\ket{\phi}\in \mathcal{H}.
 \end{align*}
 here we denote $\mathbf{Pos}_{\mathcal{H}}$ as the set of all positive semidefinite matrices. If $\rho$ is a positive semidefinite matrix with $\mathrm{tr}\rho=1,$ then $\rho$ is a state, and we denote the set of all states in $\mathcal{H}$ as $\mathcal{D}_{\mathcal{H}}.$  
  Assume $K=\sum_i \lambda_i\ket{i}\bra{i}$ is a Hermitian, then $K_{+}=\sum_{i\in \{l|\lambda_l\ge 0\}}\lambda_i\ket{i}\bra{i}$, and ${\{K\ge 0\}}=\sum_{i \in \{l|\lambda_l\ge 0\}}\ket{i}\bra{i}$.

	A quantum superoperator (channel) can be modeled by a linear map $\mathcal{E}$ from $\mathcal{D}_{\mathcal{H}_1}$ to $\mathcal{D}_{\mathcal{H}_2}$ which satisfies the following propreties: 
	\begin{itemize}
		\item[(i)] trace-preserving: $$\mathrm{tr}\mathcal{E}(\rho)=\mathrm{tr}\rho,\hspace{3mm} \forall\rho\in \mathcal{D}_{\mathcal{H}}.$$
		\item[(ii)] completely positive: assume $\mathcal{H}_0$ is an arbitrary Hilbert space, $id_{\mathcal{H}_0}$ is the identity map, 
		\begin{align*}
		[	id_{\mathcal{H}_0}\otimes\mathcal{E}](\rho)\ge 0,\hspace{3mm}\forall \rho\in \mathcal{D}_{\mathcal{H}_0\otimes\mathcal{H}_1}.
		\end{align*}
	\end{itemize}

	Here we denote the set of all superoperators on $\mathcal{H}$ as $CPTP_{\mathcal{H}}$. When no ambiguity arises, the set may be abbreviated as CPTP.
Besides, a superoperator can be characterized under the Kraus matrix form: there exists a set of matrices $\{E_k\}_k$ on $\mathcal{H}$ with $\sum_k E_k^{\dagger}E_k=\mathbb{I}$ such that 
\begin{align*}
	\mathcal{E}(\rho)=\sum_k E_k\rho E_k^{\dagger}.
\end{align*}
Here $\{E_k\}_k$ is called the Kraus matrices of $\mathcal{E}$. 
	Next we present an important type of quantum noise, the depolarizing channel.
\begin{example}
	Assume $\mathcal{H}$ is a Hilbert space with $dim(\mathcal{H})=d,$ the depolarizing channel is defined as follows,
	\begin{align*}
		\mathcal{D}_p(\rho)=(1-p)\rho+p\frac{I}{d}.
	\end{align*}
	Due to the Kraus matrix norm of $\mathcal{D}_p(\cdot)$, the state is replaced by the completely mixed state with probability $p$, while is remained unchanged with probability $1-p.$
\end{example}

	A quantum measurement can be characterized by a set of positive semi-definite operators $\mathsf{M}=\{M_i\}_{i\in \mathcal{O}}$ with $\sum_i M_i=\mathbb{I}$, $\mathcal{O}$ is the set of measurement outcomes. Here we denote the set of all POVMs as $\mathbf{POVM}.$ If the quantum state before the measurement is $\rho$, then the measurement outcome is $k$ with probability
	\begin{align*}
		p_k=\mathrm{ tr}M_k\rho.
	\end{align*}
A channel $\mathcal{M}$ generated by the POVM $\mathsf{M}=\{M_i\}_{i\in \mathcal{O}}$ is defined as 
\begin{align*}
	\mathcal{M}(\cdot)=\sum_{i\in \mathcal{O}}\mathrm{tr}[M_i(\cdot) ]\ket{i}\bra{i}.
\end{align*} 
	\subsection{Quantum divergences}
	
	In this subsection, we recall the quantum divergences that will be used here. Assume $\rho$ and $\sigma$ are two states, then the trace distance between $\rho$ and $\sigma$ is defined as follows,
	\begin{align*}
		T(\rho||\sigma)=\frac{1}{2}||\rho-\sigma||_1=\mathrm{ tr}(\rho-\sigma)_{+},
	\end{align*}

The quantum hockey-stick divergence, which can be regarded as generalizations of the trace distance, is defined as follows,
\begin{align*}
	E_{\gamma}(\rho||\sigma)=\mathrm{ tr}(\rho-\gamma\sigma)_{+},
\end{align*}	
when $\gamma=1$, the quantum hockey-stick divergence turned into the trace distance.

The quantum relative entropy between two states $\rho$ and $\sigma$ is defined as follows,
\begin{align*}
	D(\rho||\sigma)=\begin{cases}
Tr[\rho(\log\rho-\log\sigma)] \hspace{3mm}  &\supp(\rho)\subseteq \supp{\sigma}\\
+\infty\hspace{3mm}& otherwise
	\end{cases}
\end{align*}

The quantum sandwiched Renyi-$\alpha$ relative entropies is another quantity to address the differences between two states $\rho$ and $\sigma$ \cite{Müller2013}. When $\alpha>1$, it is defined as follows 
\begin{align*}
	\tilde{D}_{\alpha}(\rho||\sigma)=\begin{cases}
		\frac{1}{\alpha}\log Q_{\a}(\rho||\sigma),\hspace{3mm}  &\supp(\rho)\subseteq \supp{\sigma},\\
		+\infty\hspace{3mm}& otherwise,
	\end{cases}
\end{align*}
here $Q_{\a}(\rho||\sigma)=\mathrm{ tr}(\sigma^{\frac{1-\a}{2\a}}\rho\sigma^{\frac{1-\a}{2\a}})^{\a}$, when $\a\rightarrow 1$, the quantum sandwiched Renyi-$\alpha$ relative entropy tends to the quantum relative entropy. The sandwiched Renyi-$\a$ relative entropy satisfies additivity property \cite{Müller2013} and the data-processing property when $\alpha\ge 1$ \cite{beigi2013}. 

Assume there are two quantum states $\rho$ and $\sigma$, Alice is to guess which state is with the use of two-outcome POVM $M=\{(M,I-M)|I\ge M\ge 0\}$, here $M$ and  $I-M$ correponds to $\rho$ and $\sigma$, respectively. There are two types of mistake, 
\begin{itemize}
	\item \emph{Type I Error:} the outcome is 1 while the state is $\rho$,
	\begin{align*}
		\alpha(M)=\mathrm{tr}\rho(I-M).
	\end{align*}
	\item \emph{Type II Error:} the outcome is 0 while the state is $\sigma$,
	\begin{align*}
		\beta(M)=\mathrm{tr}\sigma M.
	\end{align*}
\end{itemize}

Asymmetry  quantum hypothese testing is to minimize the probability of \emph{Type II Error} under a bounded \emph{Type I Error}. Specificially, 
\begin{align*}
	\beta_{\vartheta}(\rho||\sigma)=&\min_M\mathrm{tr}M\sigma\\
	\textit{s. t.}\hspace{3mm}&\mathrm{ tr}(I-M)\rho\le \vartheta,\\
	&0\le M\le I.
\end{align*}

Next we recall the definition of the sample complexity of asymmetric binary hypothesis testing. Assume $\vartheta,\varphi\in [0,1]$, and let $\rho$ and $\sigma$ be two states. The sample complexity $ASC(\r,\s,\vartheta,\varphi)$ is defined as follows,
\begin{align*}
	ASC(\r,\s,\vartheta,\varphi)=\inf\{n\in\mathbb{N}|\beta_{\vartheta}(\r^{\otimes n}||\s^{\otimes n})\le \varphi\}.
\end{align*}

At last, we recall the definiton of symmetric quantum hypothesis testing. Assume there are two states $\rho$ and $\sigma$ with the probability of occurrence $p$ and $q$ with $p+q=1$, respectively, then the optimal error probability of symmetric quantum hypothesis testing is 
\begin{align}
	p_e(\rho,\sigma,p,1-p)=&\min_{I\ge M\ge 0}\{p\mathrm{tr}M\rho+q\mathrm{ tr}(I-M)\sigma\}\\
	=&\frac{1}{2}(1-||p\rho-q\sigma||_1),\label{sypt}
\end{align}
the last equality is due to the Helstrom-Holevo theorem \cite{helstrom1969quantum,holevo1973statistical}.

 The sample complexity of symmetric binary hypothesis testing between two states $\rho$ and $\sigma$ with prior probabilities $p$ and $q$, respectively, can be defined as follows \cite{cheng2025invitation},
\begin{align}
	SC(\alpha,\r,\s,p,q)=\min \{n\in\mathbb{N}|\frac{1}{2}(1-||p\rho^{\otimes n}-q\sigma^{\otimes n}||_1)\le \alpha\}. \label{scsht}
\end{align}
Assume $(\r_i)_{i=1}^m$ is a turple of quantum states with its prior probabilities $(p_i)_{i=1}^m$, correspondingly, then the sample complexity of symmetrical $m$-outcomes hypothesis testing for $((\r_i)_{i=1}^m,(p_i)_{i=1}^m)$ is defined as \cite{cheng2025invitation}
\begin{align*}
	SC(\a,(\rho_i)_{i=1}^m,(p_i)_{i=1}^m)=\inf\{n\in \mathbb{N}|p_e((\rho_i)_{i=1}^m,(p_i)_{i=1}^m)\le\a\},
	\end{align*}
	here \begin{align*}
	p_e((\rho_i)_{i=1}^m,(p_i)_{i=1}^m)=\inf_{\Lambda_1^{(n)},\cdots,\Lambda_m^{(n)}}\sum_{i=1}^m p_i\mathrm{ tr}[(I-\Lambda_i^{(n)})(\mathcal{A}(\rho_i))^{\otimes n}],
	\end{align*}
	where $\Lambda_i,$ $i=1,2,\cdots,m$ takes over all the semidefinite operators with $\sum_{i=1}^m\Lambda_i^{(n)}=I^{\otimes n}$.

Recently, the authors in \cite{cheng2025invitation} obtained bounds of sample complexity of asymetric hypothesis testing and symetric hypothesis testing for a couple of states.
\begin{Lemma}\cite{cheng2025invitation}\label{asyqhtc}
	Assume $\e,\d\in(0,1)$, and $\rho$ and $\s$ are states. Let $\gamma>1$ such that $\tilde{D}_{\gamma}(\rho||\sigma)<\infty$ and $\tilde{D}_{\gamma}(\sigma||\rho)<\infty$. Then
	\begin{align*}
		\max\{\sup_{\a\in(1,\gamma]}\left(\frac{\ln\frac{(1-\e)^{\a^{'}}}{\d}}{\tilde{D}_{\alpha}(\rho||\sigma)}\right),\sup_{\a\in(1,\gamma]}\left(\ln\frac{\frac{(1-\d)^{\a^{'}}}{\e}}{\tilde{D}_{\alpha}(\sigma||\rho)}\right)\}\le ASC(\r,\s,\e,\d)\le \min\{\lceil\inf_{\a\in (0,1)}\left(\frac{\ln\frac{\e^{\a^{'}}}{\delta}}{D_{\a}(\r||\s)}\right)\rceil,\lceil\inf_{\a\in (0,1)}\left(\frac{\ln\frac{\d^{\a^{'}}}{\e}}{D_{\a}(\r||\s)}\right)\rceil\},
	\end{align*}
	where $\a^{'}=\frac{\a}{1-\a}.$
\end{Lemma}
\begin{Lemma}\label{syqhtc}\cite{cheng2025invitation}
	Assume $\alpha\in [0,pq]$, then for non-orthogonal states $\rho$ and $\sigma$,
	\begin{align*}
		\max\{\frac{\ln(\frac{pq}{\alpha})}{-\ln F(\rho,\sigma)},\frac{1-\frac{\a(1-\a)}{pq}}{d^2_B(\rho,\sigma)}\}\le SC_{(\rho,\sigma)}(\alpha,p,q)\le \lceil\frac{2\ln\frac{\sqrt{pq}}{\alpha}}{-\ln F(\rho,\sigma)}\rceil.
	\end{align*}
	Here $d_B(\rho,\sigma)=\sqrt{2(1-\sqrt{F(\rho,\sigma)})}.$
\end{Lemma}

\begin{Lemma}\label{symqhtc}\cite{cheng2025invitation}
Assume $\e,\d\in [0,1],$ and $\r_m,\in\mathcal{D}(\mathcal{H}),$ $m=1,2,\cdots,k,$ and  $(\rho_m)_{m=1}^k$ are a tuple of states with the prior probabilities $(p_m)_{m=1}^k$. Then,
\begin{align*}
	\max_{m\ne \tilde{m}}\frac{\ln\left(\frac{p_mp_{\tilde{m}}}{(p_m\e+\e p_{\tilde{m}})}\right)}{-\ln F(\rho_m,\rho_{\hat{m}})}\le SC_{((\rho_i)_{i=1}^m,(p_i)_{i=1}^m)}^{\e}\le \left\lceil\max_{m\ne \tilde{m}}\frac{2\ln\left(\frac{M(M-1)\sqrt{p_m}\sqrt{p_{\tilde{m}}}}{2\e}\right)}{-\ln F(\rho_m,\rho_{\tilde{m}}) }\right\rceil.
\end{align*}
\end{Lemma}

\subsection{Quantum Local Differential Privacy}
\indent Here we first review the definition of quantum local differential privacy, then we present the main concept addressed here, quantum probabilistically local differential privacy, finally we show the definition of the sample complexity probabilistically privatized contraction coefficients in terms of the hockey-stick divergence.
	\begin{definition}\cite{hirche2023quantum}
	Suppose $\mathcal{A}$ is a quantum superoperator on a quantum system $\mathcal{H}$, let $\epsilon$ and $\delta$ be some nonnegative numbers. Then $\mathcal{A}$ is $(\epsilon,\delta)$-local differential privacy [($\epsilon,\delta$)-QLDP] if for any two states $\rho$ and $\sigma$ on $\mathcal{H}$, and for any subset $S\in \mathcal{O}$, we have 
	\begin{align}
		\sum_{k\in S}\mathrm{tr}(M_k\mathcal{E}(\rho))\le e^{\epsilon}\sum_{k\in S}\mathrm{tr}(M_k\mathcal{E}(\sigma))+\delta.\label{f1}
	\end{align}
\end{definition}
\begin{definition}
Assume	$\mathcal{A}$ is a quantum superoperator, and $\epsilon$ and $\delta$ are nonnegative, then $\mathcal{A}$ is $(\epsilon,\delta)$-QPrLDP if
for any couple of quantums states $\rho$ and $\sigma$ such that
\begin{align*}
	\mathrm{tr}\mathcal{A}(\rho)\{\mathcal{A}(\rho)-e^{\e}\mathcal{A}(\sigma)>  0\}\le \delta.
\end{align*}

\end{definition}

\begin{remark}
As $\{M\ge 0\}\ge \{M>0\},$ a method to obtain the upper bound of $\mathrm{tr}\mathcal{A}(\rho)\{\mathcal{A}(\rho)-e^{\e}\mathcal{A}(\s)> 0\}$ is to analyse the quantity $\mathrm{tr}\mathcal{A}(\rho)\{\mathcal{A}(\rho)-e^{\e}\mathcal{A}(\sigma)\ge 0\}$. 
	In Lemma \ref{lds}, we present some properties of $\mathrm{tr}\rho\{\rho-t\sigma\ge 0\}$. This quantity and its variants have been wildely considered to study the quantum hypothesis testing \cite{ogawa2000,nagaoka2007information}, randomness extraction \cite{tomamichel2013hierarchy} and other entanglement tasks \cite{datta2014second}.

\end{remark}

In this manuscript, we denote the set of all $(\epsilon,\delta)$-QPrLDP as
\begin{align}\label{edqprldp}
	\mathcal{Q}^{\e,\delta} =\{\mathcal{N}\in CPTP|\mathrm{tr}\mathcal{N}(\rho)\{\mathcal{N}(\rho)-e^{\e}\mathcal{N}(\sigma)\ge 0\}\le \delta,\forall\rho,\sigma\in\mathcal{D}(\mathcal{H})\}.
\end{align}

Next to show the effects of $\mathcal{Q}^{\e.\d}$ on the sample complexity of asymmetric and symmetric hypothesis testing for a couple of states $(\r,\s)$, we introduce the following concepts.
\begin{definition}
	Assume $\e,\d\in [0,1]$, $\r$ and $\s$ are two states. The sample complexity of probabilistically privatized asymmetic binary quantum hypothesis testing is defined as follows,
	\begin{align*}
		ASC^{\e,\d}(\r,\s,\vartheta,\varphi)=\inf_{\mathcal{A}\in \mathcal{Q}^{\e,\d}}\{n\in\mathbb{N}|\beta_{\vartheta}((\mathcal{A}(\r))^{\otimes n}||(\mathcal{A}(\s))^{\otimes n})\le \varphi\},
	\end{align*}
	where the infimum takes over all the quantum superoperators $\mathcal{A}\in \mathcal{Q}^{\e,\d}.$
\end{definition}
 \begin{definition}
 	Assume $\e,\d\in [0,1],$ and $\r_i\in\mathcal{D}(\mathcal{H}),$ $i=1,2,\cdots,m.$	The sample complexity of probabilistically privatized multiple hypothesis testing for a tuple of states $(\rho_i)_{i=1}^m$ with the prior probabilities $(p_i)_{i=1}^m$ to achieve at most $\alpha$ error probability is defined as
 	\begin{align*}
 		SC^{\e,\d}(\a,(\rho_i)_{i=1}^m,(p_i)_{i=1}^m)=\inf_{\mathcal{A}\in\mathcal{Q}^{\e,\d}}\{n\in \mathbb{N}|p_e((\rho_i)_{i=1}^m,(p_i)_{i=1}^m)\le\a\},
 	\end{align*}
 	here \begin{align*}
 		p_e((\rho_i)_{i=1}^m,(p_i)_{i=1}^m)=\inf_{\Lambda_1^{(n)},\cdots,\Lambda_m^{(n)}}\sum_{i=1}^m p_i\mathrm{ tr}[(I-\Lambda_i^{(n)})(\mathcal{A}(\rho_i))^{\otimes n}],
 	\end{align*}
 	where $\Lambda_i,$ $i=1,2,\cdots,m$ takes over all the semidefinite operators with $\sum_{i=1}^m\Lambda_i^{(n)}=I^{\otimes n}$. 
 \end{definition}
	
The hockey-stick divergence $E_{\gamma}(\cdot||\cdot)$ plays a key role in quantum information processing tasks \cite{hirche2023quantum,nuradha2025measured,Frenkel2023integralformula,hirche2024quantum}, which satisfies the data-processing property,
\begin{align*}
	E_{\gamma}(\rho||\sigma)\ge E_{\gamma}(\mathcal{N}(\r)||\mathcal{N}(\sigma)), 
\end{align*}	
where $\mathcal{N}$ is a channel, both $\rho$ and $\sigma$ are quantum states. From the above inequality, there exists a function $\eta(\cdot)$ on $\mathcal{N}$ such that 
\begin{align*}
1\ge	\eta_{E_{\gamma}}(\mathcal{N})\ge\sup_{\rho,\sigma} \frac{E_{\gamma}(\mathcal{N}(\rho)||\mathcal{N}(\sigma))}{E_{\gamma}(\rho||\sigma)}.
\end{align*}

The probabilistically privatized contraction coefficients in terms of the hockey-stick divergence is defined as follows,
	\begin{align*}
		\eta_{E_{\gamma}}^{\e,\delta}=\sup_{\substack{\mathcal{N}\in \mathcal{Q}^{\e,\delta}}}\sup_{\rho,\sigma}\eta_{E_{\gamma}}(\mathcal{N},\rho,\sigma)=\sup_{\substack{\mathcal{N}\in \mathcal{Q}^{\e,\delta}, \\\rho,\sigma\in \mathcal{D}(\mathcal{H}),\\E_{\gamma}(\rho||\sigma)\ne 0 }}\frac{E_{\gamma}(\mathcal{N}(\rho)||\mathcal{N}(\sigma))}{E_{\gamma}(\rho||\sigma)},
	\end{align*} 
	where $	\mathcal{Q}^{\e,\delta}$ is defined in (\ref{edqprldp}).

		\section{Properties of QPrLDP}\label{III}
	In this section, we consider the properties of $(\e,\delta)$-QPrLDP. First, we present the properties of $(\epsilon,\delta)$-QPrLDP under tensor product and unitary operations, then we present a counterexample to show that the set of $(\epsilon,\delta)$-QPrLDP is not convex, which also indicates that the set of $(\epsilon,\delta)$-QPrLDP is not  closed under a generic noisy operation.

	\begin{Theorem}
		The set $(\epsilon,\delta)$-QPrDP satisfies the following properties,
		\begin{itemize}
			\item[(i)] 	Assume $\mathcal{A}_1$ and $\mathcal{A}_2$ are $(\epsilon_1,\delta_1)$-QPrLDP and $(\epsilon_2,\delta_2)$-QPrLDP, respectively, then $\mathcal{A}_2\otimes\mathcal{A}_1$ is $(\epsilon_1+\epsilon_2,\min(\lambda_{max}(\mathcal{A}_1)\lambda_{max}(\mathcal{A}_2)(\delta_1+\delta_2),1))$-QPrLDP on such product states, here $\lambda_{max}(\mathcal{A})=\sup_{\rho}\frac{\lambda_{max}(\mathcal{A}(\r))}{\lambda_{min}(\mathcal{A}(\r))}.$
			\item[(ii)] Assume $\mathcal{A}$ is $(\epsilon,\delta)$-QPrLDP, and $\mathcal{U}(\cdot)=U\cdot U^{\dagger}$ is a unitary operation, then $\mathcal{U}\circ\mathcal{A}$ is also $(\epsilon,\delta)$-QPrLDP.
		\end{itemize}	
	\end{Theorem}
	\begin{proof}
		\begin{itemize}
			\item[(i)]
			As $\rho_1,\sigma_1,\rho_2$ and $\sigma_2$ are quantum states, $\mathcal{A}_1$ and $\mathcal{A}_2$ are $(\epsilon_1,\delta_1)$-QPrLDP and $(\epsilon_2,\delta_2)$-QPrLDP, respectively, then 
			\begin{align*}
				\mathrm{tr}\mathcal{A}_1(\rho_1)(\{\mathcal{A}_1(\rho_1-e^{\epsilon_1}\sigma_1)> 0\})\le \delta_1,\\
				\mathrm{tr}\mathcal{A}_2(\rho_2)(\{\mathcal{A}_2(\rho_2-e^{\epsilon_2}\sigma_2)> 0\})\le \delta_2,
			\end{align*}
			hence, 
			
			\begin{align*}
				&\mathrm{tr}(\mathcal{A}_1(\rho_1)\otimes\mathcal{A}_2(\rho_2)\{\mathcal{A}_1(\rho_1)\otimes\mathcal{A}_2(\rho_2)-e^{\epsilon_1+\epsilon_2}\mathcal{A}_1(\sigma_1)\otimes\mathcal{A}_2(\sigma_2)> 0\})\\
				=&\mathrm{tr}(\mathcal{A}_1(\rho_1)\otimes\mathcal{A}_2(\rho_2)\{\mathcal{A}_1(\rho_1)\otimes\mathcal{A}_2(\rho_2)-e^{\epsilon_1}\mathcal{A}_1(\sigma_1)\otimes \mathcal{A}_2(\rho_2)+e^{\epsilon_1}\mathcal{A}_1(\sigma_1)\otimes \mathcal{A}_2(\rho_2)-e^{\epsilon_1+\epsilon_2}\mathcal{A}_1(\sigma_1)\otimes\mathcal{A}_2(\sigma_2)> 0\})\\
				\le&\frac{\lambda_{max}(\mathcal{A}_1(\r_1))\lambda_{max}(\mathcal{A}_2(\rho_2))}{\lambda_{min}(\mathcal{A}_1(\r_1))\lambda_{min}(\mathcal{A}_2(\rho_2))}[\mathrm{tr}\mathcal{A}_1(\r_1)\otimes\mathcal{A}_2(\r_2)\{\mathcal{A}_1(\rho_1)\otimes\mathcal{A}_2(\rho_2)-e^{\epsilon_1}\mathcal{A}_1(\sigma_1)\otimes \mathcal{A}_2(\rho_2)> 0\}\\
				+&\mathrm{tr}\mathcal{A}_1(\r_1)\otimes\mathcal{A}_2(\r_2)\{e^{\epsilon_1}\mathcal{A}_1(\sigma_1)\otimes \mathcal{A}_2(\rho_2)-e^{\epsilon_1+\epsilon_2}\mathcal{A}_1(\sigma_1)\otimes\mathcal{A}_2(\sigma_2)>0\}]\\
				\le&\frac{\lambda_{max}(\mathcal{A}_1(\r_1))\lambda_{max}(\mathcal{A}_2(\rho_2))}{\lambda_{min}(\mathcal{A}_1(\r_1))\lambda_{min}(\mathcal{A}_2(\rho_2))}[\mathrm{tr}\mathcal{A}_1(\r_1)\{\mathcal{A}_1(\rho_1)-e^{\epsilon_1}\mathcal{A}_1(\sigma_1)> 0\}
				+\mathrm{tr}\mathcal{A}_2(\r_2)\{ \mathcal{A}_2(\rho_2)-e^{\epsilon_2}\mathcal{A}_2(\sigma_2)>0\}]\\
				\le&\frac{\lambda_{max}(\mathcal{A}_1(\r_1))\lambda_{max}(\mathcal{A}_2(\rho_2))}{\lambda_{min}(\mathcal{A}_1(\r_1))\lambda_{min}(\mathcal{A}_2(\rho_2))}(\d_1+\d_2),\\
				\le&\lambda_{max}(\mathcal{A}_1)\lambda_{max}(\mathcal{A}_2)(\delta_1+\delta_2),
			\end{align*}
			Here the first inequality is due to Lemma \ref{cnwb}, the second inequality is due to Lemma \ref{l2}.
			
			\item[(ii)] As $\mathcal{A}$ is $(\epsilon,\delta)$-QPrLDP, then for arbitrary states $\rho$ and $\sigma$,
			\begin{align*}
				\mathrm{tr}\mathcal{A}(\rho)\{\mathcal{A}(\rho)-e^{\epsilon}\mathcal{A}(\sigma)\ge 0\}\le \delta,
			\end{align*}
			then let $\{\mathcal{A}(\rho)-e^{\epsilon}\mathcal{A}(\sigma)\ge 0\}=\sum\limits_i\ket{i}\bra{i}$,
			\begin{align*}
				&\mathrm{ tr}\mathcal{U}\circ\mathcal{A}(\rho)\{\mathcal{U}\circ\mathcal{A}(\rho)-e^{\epsilon}\mathcal{U}\circ\mathcal{A}(\sigma)\ge 0\}\\
				=&\mathrm{ tr}U\mathcal{A}(\rho)U^{\dagger}\{U\mathcal{A}(\rho)U^{\dagger}-e^{\epsilon}U\mathcal{A}(\sigma)U^{\dagger}\ge 0\}\\
				=&\mathrm{ tr}U\mathcal{A}(\rho)U^{\dagger}U\sum_i\ket{i}\bra{i}U^{\dagger}\\
				=&\mathrm{ tr}\mathcal{A}(\rho)\sum_i \ket{i}\bra{i}\\
				=&\mathrm{ tr}\mathcal{A}(\rho)\{\mathcal{A}(\rho)-e^{\epsilon}\mathcal{A}(\sigma)\ge 0\}\le \delta,
			\end{align*}
			hence, we finish the proof.
		\end{itemize}
	\end{proof}
	
	At last, we present an example to show that ($\epsilon,\delta$)-QPrLDP is not closed under a generic noisy operation. 
	\begin{example}\label{e2}
		
		Assume $$\rho=\begin{pmatrix}
			0.995&0\\
			0&0.005
		\end{pmatrix},\hspace{7mm}\sigma=\begin{pmatrix}
			0.4&0\\0&0.6
		\end{pmatrix},$$ let $\e=0.01$, $\delta=0.995$, then $$\mathrm{tr}\rho\{\rho-0.01\sigma> 0\}\le 0.995.$$
		
		Let $$\Lambda(\cdot)=p_1(\cdot)+p_2 U(\cdot)U^{\dagger},$$
		where $p_1=p_2=0.5$, $U=\begin{pmatrix}
			0&1\\1&0
		\end{pmatrix},$ $$\mathrm{tr}\Lambda(\rho)\{\Lambda(\rho)-0.01\Lambda(\sigma)\ge 0\}=1>0.995.$$ Hence, $\Lambda(\cdot)$ doesnot satisfy $(\e,\delta)$-QPrLDP property.
	\end{example}
	\begin{remark}
		The example \ref{e2} also tells that the set of $(\epsilon,\delta)$-QPrLDP is not convex. Assume 
		\begin{align*}
			\Lambda_1(\rho)=\rho,\hspace{8mm}\Lambda_2(\rho)=U\rho U^{\dagger}.
		\end{align*} As $\Lambda_1$ and $\Lambda_2$ are unitary channels, $\Lambda_i\in \mathcal{Q}^{0.01,0.995},$ $i=1,2.$ However, Example \ref{e2} tells us that $0.5\Lambda_1+0.5\Lambda_2$ is not in $\mathcal{Q}^{0.05,0.995}$, that is $\mathcal{Q}^{0.005,0.995}$ is not convex. More generally, $\mathcal{Q}^{\e,\delta}$ is not convex when $\delta>0.$
	\end{remark}
	\section{Relations between QLDP and QPrLDP}\label{IV}
	In this section, we will consider the relations between $(\epsilon_1,\delta_1)$-QPrLDP and $(\epsilon_2,\delta_2)$-QLDP when $\epsilon_1,\epsilon_2,\delta_1,\delta_2\ge 0$.
	\begin{Theorem}\label{t7}
		Assume $\mathcal{H}$ is a system with finite dimensions, $\mathcal{A}$ is a quantum superoperator acting on $\mathcal{H}$. If $\mathcal{A}$ is $(\epsilon,\delta)$-QPrLDP, then $\mathcal{A}$ is $(\epsilon,\delta)$-QLDP.
	\end{Theorem}
	\begin{proof}
		Based on the definition of $(\e,\d)$-QPrLDP, we have for any $\rho,\sigma\in \mathcal{D}(\mathcal{H})$,
		\begin{align*}
			\mathrm{tr}\mathcal{A}(\rho)\{\mathcal{A}(\rho)-e^{\epsilon}\mathcal{A}(\sigma)> 0\}\le\delta,
		\end{align*}
		then
		\begin{align}
			&E_{e^{\epsilon}}(\mathcal{A}(\rho)||\mathcal{A}(\sigma))\nonumber\\=&\mathrm{tr}(\mathcal{A}(\rho-e^{\epsilon}\sigma))_{+}\nonumber\\
			\le&\mathrm{tr}(\mathcal{A}(\rho-e^{\epsilon}\sigma))\{\mathcal{A}(\rho-e^{\epsilon}\sigma)> 0\}+e^{\epsilon}\mathrm{ tr}\mathcal{A}(\sigma)\{\mathcal{A}(\rho-e^{\epsilon}\sigma)> 0\}\label{qdrqpdr}\\
			=&\mathrm{tr}\mathcal{A}(\rho)\{\mathcal{A}(\rho-e^{\epsilon}\sigma)> 0\}\nonumber\\
			\le&\delta,\nonumber
		\end{align}
		hence, based on \cite{hirche2023quantum}, $\mathcal{A}$ is $(\epsilon,\delta)$-QLDP.
	\end{proof}
		\begin{remark}
		When $\mathcal{A}(\cdot)$ is $(\epsilon,0)$-QPrLDP, then if $\rho$ and $\sigma$ are two states,
		\begin{align*}
			&0= \mathrm{tr}\mathcal{A}(\rho)\{\mathcal{A}(\rho)-e^{\e}\mathcal{A}(\sigma)> 0\}
			=E_{e^{\e}}(\mathcal{A}(\rho)||\mathcal{A}(\sigma))+\mathrm{ tr}e^{\e}\mathcal{A}(\sigma)\{\mathcal{A}(\rho)-e^{\e}\mathcal{A}(\sigma)> 0\}=0\\
			\Longrightarrow& E_{e^{\e}}(\mathcal{A}(\rho)||\mathcal{A}(\sigma))=0.
		\end{align*}
		Hence,	 $\mathcal{A}$ is $(\epsilon,0)$-QLDP.  Besides, when $\mathcal{A}(\cdot)$ is $(\epsilon,0)$-QLDP, based on the relations between QLDP and QPrLDP, $\mathcal{A}(\cdot)$ is also $(\epsilon,0)$-QPrLDP. That is, $(\epsilon,0)$-QPrLDP $=(\epsilon,0)$-QLDP.
	\end{remark}
	
	Generally, if $\mathcal{A}$ is $(\epsilon,\delta)$-QLDP, as the second term in (\ref{qdrqpdr}) is not 0 in general, $\mathcal{A}$ is not $(\epsilon,\delta)$-QPrLDP. Neverthesless, when $(\epsilon_0,\delta_0)$ is not equal to $(\epsilon_1,\delta_1)$, a superoperator with $(\epsilon_0,\delta_0)$-QLDP is $(\epsilon_1,\delta_1)$-QPrLDP. 
	\begin{Theorem}
		Assume $\epsilon\ge \epsilon^{'}\ge 0,\delta\ge 0$, if $\mathcal{A}$ is $(\epsilon,\delta)$-QLDP, then $\mathcal{A}$ is $(\epsilon^{'},\frac{\delta}{(1-e^{\epsilon-\epsilon^{'}})})$-QPrLDP.
	\end{Theorem} 
	\begin{proof}
		As $\mathcal{A}$ is $(\epsilon,\delta)$-QLDP, for any $\rho,\sigma\in \mathcal{D}(\mathcal{H})$,
		\begin{align*}
			\mathrm{tr}[\mathcal{A}(\rho)-e^{\epsilon}\mathcal{A}(\sigma)]_{+}\le \delta.
		\end{align*}
		
		Assume $\{M_k\}$ is any POVM, $\rho$ and $\sigma$ are two states, $\mathrm{ tr}M_k(\mathcal{A}(\rho)-e^{\epsilon}\mathcal{A}(\sigma))\le \delta$, then let $T=\{k|\mathrm{ tr}[M_k\mathcal{A}(\rho)- e^{\epsilon^{'}}M_k\mathcal{A}(\sigma)]> 0\}$,
		\begin{align*}
			\sum_{k\in T}\mathrm{ tr}\mathcal{A}(\rho)M_k\le& \sum_{k\in T}e^{\epsilon}\mathrm{ tr}M_k\mathcal{A}(\sigma)+\delta\\
			\le&\sum_{k\in T}e^{\epsilon-\epsilon^{'}}\mathrm{ tr}M_k\mathcal{A}(\rho)+\delta,
		\end{align*}
		that is, $\sum_{k\in T}\mathrm{tr}M_k\mathcal{A}(\rho)\le \frac{\delta}{(1-e^{\epsilon-\epsilon^{'}})}.$ As when $M_k=\{\mathcal{A}(\rho)-e^{\e^{'}}\mathcal{A}(\s)>0\},$ $\mathrm{tr}M_k(\mathcal{A}(\r)-e^{\e^{'}}\mathcal{A}(\s))>0,$  Hence, based on the definition of QPrLDP, we finish the proof.
	\end{proof}
	
	\section{Examples}\label{V}
	In this section, we will present when a depolarizing noise is $(\epsilon,\delta)$-QPrLDP  under a serious of scenarios.
	\subsection{Global Depolarizing Noise}
	A depolarizing channel is defined as follows,
	\begin{align*}
		\mathcal{D}_p(\rho)=(1-p)\rho+p\frac{I}{d},
	\end{align*}
	here $p\in [0,1]$ and $d$ is the dimension of the system. 
	\begin{Lemma}\label{the gd}
		Assume $p\in [0,1]$, and $\epsilon\ge 0$, then 
		\begin{align*}
			\mathrm{ tr}\mathcal{D}_p(\rho)\{\mathcal{D}_p(\rho)-e^{\epsilon} \mathcal{D}_p(\sigma)> 0\}\le f(p),
		\end{align*}
		here $f(p)=(1-p)\mathrm{tr}\rho\{\rho-e^{\epsilon}\sigma\ge 0\}+p,$ that is, $\mathcal{D}_p(\cdot)$ is $(e^{\epsilon},f(p))$-QPrLDP.
	\end{Lemma}
	\begin{proof}
		Based on the definition of $\mathcal{D}_p(\cdot)$,
		\begin{align*}
			&\mathrm{tr}\mathcal{D}_p(\rho)\{\mathcal{D}_p(\rho)-e^{\e}\mathcal{D}_p(\sigma)>0\}\\\le&\mathrm{tr}\mathcal{D}_p(\rho)\{\mathcal{D}_p(\rho)-e^{\epsilon}\mathcal{D}_p(\sigma)\ge 0\}\\
			=&\mathrm{tr}[(1-p)\rho+p\frac{I}{d}]\{(1-p)(\rho-e^{\epsilon}\sigma)+p(1-e^{\epsilon})\frac{I}{d}\ge 0\}\\
			\le&\mathrm{tr}[(1-p)\rho+\frac{pI}{d}]\{\rho-e^{\epsilon}\sigma\ge 0\}\\
			=&(1-p)\mathrm{tr}\rho\{\rho-e^{\epsilon}\sigma\ge 0\}+\frac{p}{d}\mathrm{tr}\{\rho-e^{\epsilon}\sigma\ge0\}\\
			\le &(1-p)\mathrm{tr}\rho\{\rho-e^{\epsilon}\sigma\ge 0\}+p,
		\end{align*}
		here the first inequality is due to that $\{(1-p)(\rho-e^{\epsilon}\sigma)+p(1-e^{\epsilon})\frac{I}{d}\ge 0\}\ge \{(1-p)(\rho-e^{\epsilon}\sigma)\ge0\}$, the second inequality is due to that $dim\{\rho-e^{\epsilon}\sigma\ge 0\}\le d.$
		
	\end{proof}

	\begin{Corollary}
		Assume $\rho$ and $\sigma$ are two states with $\frac{1}{2}||\rho-\sigma||_1\le \chi$, then $\mathcal{D}_p$ is $(\epsilon,\delta)$-QPrLDP with 
		\begin{align*}
			\delta\ge (1-p)\chi+e^{\epsilon}(1+\chi)(1-p)+p,
		\end{align*}
		if $\mathcal{A}=\circ_{i=1}^n\mathcal{D}_{p_i}$, then $\mathcal{A}$ is $(\epsilon,\delta)$-QPrLDP with
		\begin{align*}
			\delta\ge & \Pi_{i=1}^n(1-p_i)\chi+e^{\e}(1+\chi)\Pi_{i=1}^n (1-p_i)+1-\Pi_{i=1}^n (1-p_i).
		\end{align*}
	\end{Corollary}
	\begin{proof}
		Based on Lemma \ref{the gd}, we have
		\begin{align*}
			&\mathrm{tr}\mathcal{D}_p\{\mathcal{D}_p(\rho)-e^{\epsilon}\mathcal{D}_p(\sigma)> 0\}\\
			\le& (1-p){D}_{e^{\epsilon}}(\rho||\sigma)+p\\
			\le&(1-p) E_{e^{\epsilon}}(\rho|||\sigma)+(1-p)\inf_{s\in [0,1]}e^{\epsilon(2-s)}\mathrm{tr}\rho^s\sigma^{1-s}+p\\
			\le& (1-p)\chi+e^{\epsilon}(1+\chi)(1-p)+p.
		\end{align*}
		The last inequality is due to $E_{e^{\e}}(\rho||\sigma)\le \frac{1}{2}||\rho-\sigma||_1\le \chi$ \cite{hirche2023quantum} and Lemma \ref{lnain}.
		
		As $\frac{1}{2}||\rho-\sigma||_1\le \chi$, then 
		\begin{align*}
			&\frac{1}{2}||\mathcal{D}_p(\rho)-\mathcal{D}_p(\sigma)||_1\\
			\le&\frac{1}{2}||(1-p)\rho+\frac{pI}{d}-(1-p)\sigma-\frac{pI}{d}||_1 \\
			=&\frac{1-p}{2}||\rho-\sigma||_1,
		\end{align*}
		If $\mathcal{A}=\circ_{i=1}^n\mathcal{D}_{p_i}$,  then $\circ_{i=1}^n \mathcal{D}_{p_i}$ can be seen as $\mathcal{D}_{1-\Pi_{i=1}^n(1-p_i)}$, hence,
		\begin{align*}
			&\mathrm{ tr}\mathcal{A}(\rho)\{\mathcal{A}(\rho)-\epsilon\mathcal{A}(\sigma)> 0\}
			\le \Pi_{i=1}^np_i\chi+e^{\e}(1+\chi)\Pi_{i=1}^n p_i+1-\Pi_{i=1}^np_i.
		\end{align*}
	\end{proof}
	
	In Fig \ref{fig1}, we present the function $\delta=f(\epsilon)$ against $\epsilon$ for the global depolarizing noise channel with $\mathcal{D}_{\epsilon}(\rho||\sigma)=0.1$ when $p=0.05$, $p=0.15$ and $p=0.25.$ 	In Fig \ref{fig2}, we plot the function on $\delta$ after $n$ compositions of the global depolarizing noise channel with $\mathcal{D}_{\epsilon}(\rho||\sigma)=0.1$ when $p=0.05$, $p=0.15$ and $p=0.25.$ 
	\begin{figure}[htbp] % 强制当前位置
		\centering
		\includegraphics[width=0.55\textwidth]{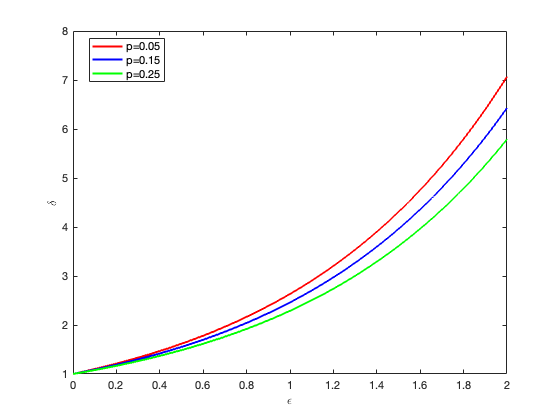} % 替换为你的图片名
		\caption{The QPrDP for the global depolarizing noise. Here $\rho$ and $\sigma$ are states with $E_1(\rho||\sigma)=\chi=0.01.$ The red line, blue line and green line denote the function $\delta=f(\epsilon)$ when $p=0.05$, $p=0.15$ and $p=0.25,$ respectively.}
		\label{fig1}
	\end{figure}
	
	\begin{figure}[htbp] % 强制当前位置
		\centering
		\includegraphics[width=0.55\textwidth]{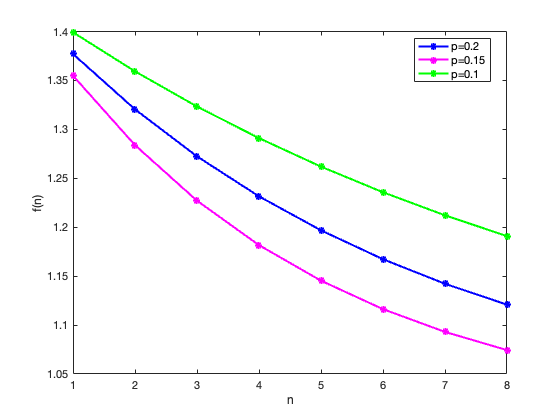} % 替换为你的图片名
		\caption{The QPrDP for the $n$ compositions of the global depolarizaing noise. Here $\epsilon=0.2,$ $\rho$ and $\sigma$ are states with $E_1(\rho||\sigma)=0.1.$ The blue line, red line and green line denote the function $\delta=f(n)$ when $p=0.1$, $p=0.15$ and $p=0.2,$ respectively.}
		\label{fig2}
	\end{figure}
	
	Alternatively, we could also derive a similar result of the bounds on $\epsilon$ in the following corollary.
	\begin{Corollary}
		Assume $\rho$ and $\sigma$ are two states with $\frac{1}{2}||\rho-\sigma||_1\le \chi$, then $\mathcal{D}_p$ is $(\epsilon,\delta)$-QPrLDP with 
		\begin{align*}
			\epsilon\le \log \frac{\delta-p-(1-p)\chi}{(1+\chi)(1-p)},
		\end{align*}
		if $\mathcal{A}=\circ_{i=1}^n\mathcal{D}_{p_i}$, then $\mathcal{A}$ is $(\epsilon,\delta)$-QPrLDP with
		\begin{align*}
			\epsilon\ge \log \frac{\delta-1+\Pi_{i=1}^np_i-\Pi_{i=1}^np_i\chi}{(1+\chi)\Pi_{i=1}^np_i}.
		\end{align*}
	\end{Corollary}
	Next we present an example when a channel generated from a depolarizing noise is $(\epsilon,\delta)$-QPrLDP.
	
	\begin{Corollary}\label{DpQPrLDP}
		When
		\begin{align*}
			p\in [0,\frac{2E_{e^{\e}}(\r||\s)}{2E_{e^{\e}}(\r||\s)+e^{\e}-1}),\hspace{6mm}	\delta_1=\max(\mathrm{tr}\rho\{\rho-e^{\e}\sigma\ge 0\},\frac{(e^{\e}-1)\mathrm{tr}\rho\{\rho-e^{\e}\sigma\ge 0\}+E_{e^{\e}}(\r||\s)}{2E_{\e}(\r||\s)+e^{\e}-1}),
		\end{align*}or $p\ge\frac{2E_{e^{\e}}(\r||\s)}{2E_{e^{\e}}(\r||\s)+e^{\e}-1}$, $\d>0,$
		then there exists a projection operator $0\le M\le I$ such that $\mathcal{D}_p\circ\mathcal{M}$ is $(\epsilon,\delta)$-QPrLDP, where $\mathcal{M}$ and $\mathcal{D}_p(\cdot)$ are defined as follows,
		\begin{align*}
			\mathcal{M}(\cdot)=&\mathrm{tr} M(\cdot) \ket{0}\bra{0}+\mathrm{tr}(I-M)(\cdot)\ket{1}\bra{1},\\
			\mathcal{D}_p(\cdot)=&(1-p)(\cdot)+\frac{p}{2}I.
		\end{align*}
		
	\end{Corollary}
	\begin{proof}
		Let $M=\{\rho-e^{\e}\sigma> 0\},$  
		\begin{align*}
			&\mathcal{D}_p\circ \mathcal{M}(\rho-e^{\e}\sigma)\\=&\mathcal{D}_p(\mathrm{ tr}(M(\rho-e^{\e}\sigma))\ket{0}\bra{0}+\mathrm{tr}(I-M)(\rho-e^{\e}\sigma)\ket{1}\bra{1})\\
			=&(1-p)[\mathrm{tr}M(\rho-e^{\e}\sigma)\ket{0}\bra{0}+\mathrm{tr}(I-M)(\rho-e^{\e}\sigma)\ket{1}\bra{1}]+\frac{p}{2}(1-e^{\e})I\\
			=&(1-p)E_{e^{\e}}(\rho||\sigma)\ket{0}\bra{0}+(1-p)(1-e^{\e}-E_{e^{\e}}(\rho||\sigma))\ket{1}\bra{1}+\frac{p}{2}(1-e^{\e})I\\
			=&[(1-p)E_{e^{\e}}(\rho||\sigma)+\frac{p}{2}(1-e^{\e})]\ket{0}\bra{0}+[(1-p)(1-e^{\e}-E_{e^{\e}}(\rho||\sigma))+\frac{p}{2}(1-e^{\e})]\ket{1}\bra{1}\\
			=&             [E_{e^{\e}}(\rho||\sigma)+\frac{p}{2}(1-e^{\e}-2E_{e^{\e}}(\rho||\sigma))] \ket{0}\bra{0}    +[p(\frac{1}{2}e^{\e}-\frac{1}{2}+E_{e^{\e}}(\r||\s))+1-e^{\e}-E_{e^{\e}}(\rho||\s)]\ket{1}\bra{1},\\
			&\mathcal{D}_p\circ\mathcal{M}(\rho)\\
			=&\mathcal{D}_p(\mathrm{tr}(M\rho)\ket{0}\bra{0}+\mathrm{tr}(I-M)\rho\ket{1}\bra{1})\\
			=&[(1-p)\mathrm{tr}(M\rho)+\frac{p}{2}]\ket{0}\bra{0}+[(1-p)\mathrm{tr}(I-M)\rho+\frac{p}{2}]\ket{1}\bra{1},
		\end{align*}
		when $ p\in [0,\frac{2E_{\e}(\r||\s)}{2E_{\e}(\r||\s)+e^{\e}-1}),$ $$E_{e^{\e}}(\rho||\sigma)+\frac{p}{2}(1-e^{\e}-2E_{e^{\e}}(\rho||\sigma))> E_{e^{\e}}(\rho||\sigma)+\frac{E_{\e}(\rho||\sigma)}{2E_{\e}(\rho||\sigma)+e^{\e}-1}\times(1-2E_{\e}(\rho||\sigma)-e^{\e})\ge 0,$$
		\begin{align*}
			&\mathrm{tr}\mathcal{D}_p\circ \mathcal{M}(\rho)\{\mathcal{D}_p\circ\mathcal{M}(\rho-e^{\e}\sigma) >0\}\\
			=&(1-p)\mathrm{tr}M\rho+\frac{p}{2}\\
			\le& \delta.
		\end{align*}
		When $p\ge \frac{2E_{\e^{\e}(\rho||\sigma)}}{e^{\e}-1+2E_{e^{\e}}(\r||\s)},$ $(1-p)E_{\e}(\rho||\sigma)+\frac{p}{2}(1-e^{\e})\le 0$ and $(1-p)(1-e^{\e}-E_{e^{\e}}(\rho||\sigma))+\frac{p}{2}(1-e^{\e})\le 0,$ then $\{\mathcal{D}_p\circ\mathcal{M}(\rho-e^{\e}\sigma)> 0\}=0,$
		\begin{align*}
			\mathrm{tr}\mathcal{D}_p\circ \mathcal{M}(\rho)\{\mathcal{D}_p\circ\mathcal{M}(\rho-e^{\e}\sigma) >0\}=0\le \delta.
		\end{align*}
		Based on the definition of ($\epsilon$,$\delta$)-QPrLDP, we finish the proof of this theorem.
	\end{proof}

	\subsection{Local Depolarizing Noise}
	In the last subsection, we addressed how the quantum superoperators are affected by the global depolarizing noise. Nevertheless, it is necessary to consider each local system affected by local noise. Here we will address the local depolarizing noise with the form $\mathcal{D}_p^{\otimes k}(\cdot)$, $k$ is the number of local systems.
	\begin{Theorem}
		Assume $p\in [0,1]$, $\gamma\ge 1$, and $\rho,\sigma\in \mathcal{D}(\mathcal{H}_{A_1A_2\cdots A_k})$ with $\frac{1}{2}||\rho-\sigma||_1\le\eta,$
		\begin{align*}
		\mathrm{tr}\mathcal{D}_p^{\otimes k}(\rho)\{\mathcal{D}_p^{\otimes k}(\rho)-\gamma\mathcal{D}_p^{\otimes k}(\sigma)> 0\}\le f(\eta).
		\end{align*}
		here $f(\eta)=(1-p)^{2k}(1-\eta^2)+\frac{2(1-\eta^2)}{d^n}[1-p^k-(1-p)^k]+\frac{2p^k-p^{2k}-p^k(1-p)^k}{d^k}+(1-p^m-(1-p)^m)^2.$
	\end{Theorem}
	\begin{proof}
		Note that $\mathcal{D}_p^{\otimes k}(\rho)$ can be written as 
		\begin{align*}
			\mathcal{D}_p^{\otimes k}(\rho)=&(1-p)^k\rho+\frac{p^k}{d^k}I\\
			+&\sum_{m=1}^{k-1}p^m(1-p)^{k-m}\left(\mathrm{tr}_{A_1A_2\cdots A_m}\rho\otimes I_{A_1A_2\cdots A_m}\right)\\
		\end{align*}
		here all the subsystems $A_1A_2\cdots A_m$ takes over all the $m$ subsystems of the whole system, then
		
		\begin{align*}
			&\mathrm{tr}\mathcal{D}_p^{\otimes k}(\rho)\{\mathcal{D}_p^{\otimes k}(\rho)-\mathcal{D}_p^{\otimes k}(\sigma)> 0\}\\
			\le&\mathrm{tr}\mathcal{D}_p^{\otimes k}(\rho)\{\mathcal{D}_p^{\otimes k}(\rho)-\mathcal{D}_p^{\otimes k}(\sigma)\ge 0\}\\
			=&\mathrm{tr}(1-p)^k\rho+\frac{p^k}{d^k}I
			+\sum_{m=1}^{k-1}p^m(1-p)^{k-m}\left(\mathrm{tr}_{A_1A_2\cdots A_m}\rho\otimes I_{A_1\cdots A_m}\right)\{(1-p)^k(\rho-\sigma)\\+&\sum_{m=1}^{k-1}p^m(1-p)^{k-m}\left(\mathrm{tr}_{A_1A_2\cdots A_m}(\rho-\sigma)\otimes I_{A_1A_2\cdots A_m}\right)\ge 0\}\\
			=&D_{\gamma}(\phi_p(\rho)||\phi_p(\sigma))+\frac{p^k }{d^k}\mathrm{tr}\{(1-p)^k(\rho-\sigma)+\sum_{m=1}^{k-1}p^m(1-p)^{k-m}\left(\mathrm{tr}_{A_1A_2\cdots A_m}(\rho-\sigma)\otimes I_{A_1A_2\cdots A_m}\right)\ge 0\}\\
			\le& E_{\gamma}(\phi_p(\rho)||\phi_p(\sigma))+\inf_{s\in [0,1]}\gamma^{2-s}\mathrm{tr}\phi^s_p(\rho)\phi_p^{1-s}(\sigma)\\
			\le&(1-p)^{2k}(1-\eta^2)+\frac{2(1-\eta^2)}{d^n}[1-p^k-(1-p)^k]+\frac{2p^k-p^{2k}-p^k(1-p)^k}{d^k}+(1-p^m-(1-p)^m)^2.
		\end{align*}
		
		In the second equality, ${\phi}_p(\rho)=(1-p)^k\rho+\sum_{m=1}^{k-1}p^m(1-p)^{k-m}\left(\mathrm{tr}_{A_1A_2\cdots A_m}\rho\otimes I_{A_1\cdots A_m}\right)+\frac{p^kI}{d^k}$, and ${\phi}_p(\sigma)=(1-p)^k\sigma+\sum_{m=1}^{k-1}p^m(1-p)^{k-m}\left(\mathrm{tr}_{A_1A_2\cdots A_m}\sigma\otimes I_{A_1\cdots A_m}\right)+\frac{p^kI}{d^k},$ in the last inequality, we apply Lemma \ref{dtf}.
	\end{proof}

	\section{The sample complexity of quantum probabilistically differentially private}\label{VI}

	Assume $\rho$ and $\sigma$ are two states, the operational interpretation of $\mathrm{tr}\rho\{\rho-e^{\e}\sigma>0\}\le \delta$ means that the acceptance probability under $\rho$ of the quantum Neyman-Pearson test at threshold $e^{\e}$ is less than $\delta$, it also means that the probability of the spectral violation event happens on the system $\r$ is less than $\delta$. Hence, if $\rho$ and $\s$ satisfy $\mathrm{tr}\rho\{\rho-e^{\e}\sigma>0\}\le \delta,$ we can denote the property as $(\e,\d)$-quantum spectral probabilistic indistinguishability (QSPrI). Here we first consider the sample complexity of $(\e,\d)$-QSPrI.
	
	\begin{definition}
		Assume $\rho$ and $\sigma$ are two states, the sample complexity $n^{*}(\rho,\sigma)$ of $(\epsilon,\delta)$-QSPrI is defined as follows,
		\begin{align*}
			n^{*}(\rho,\sigma,\epsilon,\delta)=\inf\{n\in\mathbb{N}|\mathrm{tr}\rho^{\otimes n}\{\rho^{\otimes n}-e^{\epsilon}\sigma^{\otimes n}\}> 0\}\le \delta.
		\end{align*}
	\end{definition}

Based on the definition of $(\e,\d)$-QSPrI, it is hard to obtain the exact values of $n^{*}(\r,\s,\e,\d)$ for generic states $\r$ and $\s$. Here we will present the upper bounds of $n^{*}(\r,\s,\e,\d)$ for arbitary couple of states $(\r,\s)$.
	\begin{Theorem}\label{tsc}
		Assume $\epsilon\ge 0$, $\delta\in [0,1]$, $\rho$ and $\sigma$ are two states with $\supp(\sigma)=\supp(\r),$ and $E_{e^{\e}}(\rho||\sigma)<1$, then the sample complexity $n^{*}(\rho,\sigma,\epsilon,\delta)$ satisfies the following bounds,   
		\begin{align*}
		 n^{*}(\rho,\sigma,\epsilon,\delta)\le \min( \min_{s\in (0,1)}\frac{\log(1-\delta)-\epsilon s}{\log\mathrm{tr}\rho^{1-s}\sigma^s}, \inf_{\a>1}\left\lceil\frac{\log\delta-\log[1+\frac{(\a-1)^{\a-1}}{\a^{\a}}]-\e(1-\frac{1}{\a})}{\log Q_{\a}(\rho||\s)}\right\rceil).
		\end{align*}
		Here $Q_{\a}(\r||\s)=\mathrm{ tr}(\rho^{\frac{1-\a}{2\a}}\s\r^{\frac{1-\a}{2\a}})^{\a}$, and $\lambda_{min}(M)$ is the minimal eignvalue of $M$.
	\end{Theorem}
	\begin{proof}
		Let $\gamma=e^{\e},$ then
			\begin{align*}
			\mathrm{tr}\r\{\rho-\gamma\sigma>0\}=&1-\mathrm{tr}\rho\{\rho-\gamma\sigma\le0\}\\
			\ge&1-\gamma^s\mathrm{ tr}\r^{1-s}\s^s,
		\end{align*}
		here the second inequality is due to the following,
		\begin{align*}
			\mathrm{tr}\rho^{1-s}(\gamma\sigma)^s\ge& \mathrm{ tr}\rho\{\rho-\gamma\sigma\le 0\}+\mathrm{tr}\gamma\sigma\{\rho-\gamma\sigma>0\}\\
			\ge&\mathrm{tr}\rho\{\rho-\gamma\sigma\le 0\},
		\end{align*}
	here the first inequality is based on \cite{audenaert2007}, hence,
	\begin{align*}
		\delta\ge& \mathrm{tr}\rho^{\otimes n}\{\rho^{\otimes n}-\gamma\sigma^{\otimes n}> 0\}\\		\ge&1-\gamma^s\mathrm{tr}(\rho^{\otimes n})^{1-s}(\sigma^{\otimes n})^s\\
		\ge&1-\gamma^s\mathrm{tr}^n\rho^{1-s}\sigma^s,\\
		\Longrightarrow& n\le \frac{\log(1-\delta)-\epsilon s}{\log\mathrm{tr}\rho^{1-s}\sigma^s}.
	\end{align*}
As the inequality is valid whenever $s\in (0,1)$, $n\le \min_{s\in (0,1)}\frac{\log(1-\delta)-\epsilon s}{\log\mathrm{tr}\rho^{1-s}\sigma^s}.$
		
	Then we present the other method to obtain the upper bound of $n^{*}(\rho,\sigma)$. Let $n=\inf_{\a>1}\left\lceil\frac{\log\delta-\log[1+\frac{(\a-1)^{\a-1}}{\a^{\a}}]-\e(1-\frac{1}{\a})}{\log Q_{\a}(\rho||\s)}\right\rceil,$ based on Lemma \ref{na}, 
		\begin{align*}
			&\mathrm{tr}\rho^{\otimes n}\{\rho^{\otimes n}-\gamma\sigma^{\otimes n}>0\}\\
			\le&E_{\gamma}(\rho^{\otimes n}||\sigma^{\otimes n})+\gamma\mathrm{ tr}\sigma^{\otimes n}\{\rho^{\otimes n}-\gamma\sigma^{\otimes n}\ge 0\}\\
			\le &	E_{\gamma}(\rho^{\otimes n}||\sigma^{\otimes n})+\inf_{s\in [0,1]}\gamma^{1-s}\mathrm{ tr}^n{\rho^s}{\sigma^{1-s}}\\
			\le& 	 \frac{(\a-1)^{\a-1}}{\a^{\a}}\gamma^{1-\a}\mathrm{tr}^n(\rho^{\frac{1-\a}{2\a}}\sigma\rho^{\frac{1-\a}{2\a}})^{\a}+\inf_{s\in [0,1]}\gamma^{1-s}\mathrm{ tr}^n{\rho^s}{\sigma^{1-s}}\\
			\le& \inf_{\a>1}[1+\frac{(\a-1)^{\a-1}}{\a^{\a}}]\gamma^{1-\frac{1}{\a}} Q^n_{\alpha}(\rho||\s)\le\delta,
		\end{align*}
		the second inequality is due to Lemma \ref{sandwichedgamma}. The last inequality is due to the monotonity of quantum sandwiched renyi relative entropy \cite{Müller2013} and $Q_{\a}(\rho||\s)\ge \mathrm{tr}\rho^s\sigma^{1-s}$ when $s\in (0,1).$
		Hence, we finish the proof.
	\end{proof}

		Next we show the probabilistic private contraction of the hockey-stick divergence under $(\epsilon,\delta)$-QPrLDP and provides the lower bounds, which will be applied to obtain the sample complexity for $n$-outcome symmetric hypothesis testing under $(\e,\d)$-QPrLDP.
	\begin{Theorem}\label{contractionco}
		Assume $\epsilon\ge 0$, $\gamma\in (e^{-\e},e^{\e})$, let $\rho$ and $\sigma$ be states such that $E_{\gamma}(\rho||\sigma)\ne 0,$ then
		\begin{align*}
			\eta_{E_{\gamma}}^{\e,\delta}=\sup_{\mathcal{N}\in \mathcal{Q}^{\epsilon,\delta}}\frac{E_{\gamma}(\mathcal{N}(\rho)||\mathcal{N}(\sigma))}{E_{\gamma}(\rho||\sigma)}\le f(\epsilon,\delta),
		\end{align*}
		where $f(\epsilon,\delta,\gamma)=(\gamma+1)\delta-\gamma.$ When $\gamma=1$,
		\begin{align*}
			\eta_{T}^{\epsilon,\delta}=\sup_{\mathcal{N}\in \mathcal{Q}^{\e,\delta}}\frac{T(\mathcal{N}(\rho)||\mathcal{N}(\sigma))}{T(\rho||\sigma)}\le 2\delta-1.
		\end{align*}
	\end{Theorem}
	\begin{proof}
		Let $\Pi_1=\{\rho-\gamma\sigma\ge 0\},$ $\Pi_2=\{\sigma-e^{\e}\rho> 0\}$ and $\Pi_3=\{\rho-e^{\e}\sigma> 0\}.$ Based on \cite{hirche2023quantum}, the hockey-stick divergence contraction coefficient can be expressed as
		\begin{align}
			\eta_{E_{\gamma}}^{\epsilon,\delta}=&\sup_{\mathcal{A}\in \mathcal{Q}^{\epsilon,\delta}}\sup_{\ket{\phi}\perp\ket{\psi}}E_{\gamma}(\mathcal{A}(\phi)||\mathcal{A}(\psi))\nonumber\\
			\le&\sup_{\substack{\mathrm{ tr}\rho\Pi_3\le \delta,\\
					\mathrm{tr}\sigma\Pi_2\le \delta}}E_{\gamma}(\rho||\sigma)\label{del1}
		\end{align}
		As $\Pi_2=\{\sigma-e^{\e}\rho>0\}=\{e^{-\e}\sigma-\rho>0\}$, then
		\begin{align*}
			E_{\gamma}(\rho||\sigma)=&\mathrm{tr}(\rho-\gamma\sigma)\{\rho-\gamma\sigma> 0\}\\
			=&\mathrm{tr}(\rho\{\rho-\gamma\sigma>0\})-\gamma(1-\mathrm{tr}\sigma\{\sigma-\frac{1}{\gamma}\rho\ge 0\})\\
			\le&(\gamma+1)\delta -\gamma.
		\end{align*}
		Here the inequality is due to that $\mathrm{tr}\rho\{\rho-\gamma\sigma>0\}$ is monotone on $\gamma$ \cite{nagaoka2007information}. Hence, we finish the proof.
	\end{proof}
	\begin{remark}
		Based on the proof of Theorem \ref{contractionco}, $(\gamma+1)\delta-\gamma\ge E_{\gamma}(\rho||\sigma)\ge 0$, $\delta\ge \frac{\gamma}{1+\gamma}.$ As $\gamma\in (-e^{\e},e^{\e})$, we have $\delta\ge \frac{e^{\e}}{1+e^{\e}}.$
	\end{remark}
	
	\begin{Corollary}
		Assume $\epsilon,\delta>0$, and $\gamma\in (e^{-\e},e^{\e})$, let $\rho$ and $\sigma$ be states such that $E_{\gamma}(\rho||\sigma)\ne 0$, then 
		\begin{align*}
			\eta_{E_{\gamma}}^{\epsilon,\delta}\le \min((\gamma+1)\delta-\gamma,\frac{\delta}{\gamma}+\delta-\frac{1}{\gamma}).
		\end{align*}
	\end{Corollary}
	\begin{proof}
		When $\gamma\in (e^{-\e},e^{\e})$, and $\mathcal{A}\in \mathcal{Q}^{\e,\delta},$
		\begin{align*}
			E_{\gamma}(\mathcal{A}(\rho)||\mathcal{A}(\sigma))=&\gamma E_{\frac{1}{\gamma}}(\mathcal{A}(\sigma)||\mathcal{A}(\rho))\\
			=&\gamma f(\epsilon,\delta)E_{\frac{1}{\gamma}}(\sigma||\rho)\\
			=&\gamma((\frac{1}{\gamma}+1)\delta-\frac{1}{\gamma})E_{\frac{1}{\gamma}}(\sigma||\rho)\\
			=&((\frac{1}{\gamma}+1)\delta-\frac{1}{\gamma})E_{{\gamma}}(\rho||\sigma).\\
			\Longrightarrow&\eta_{E_{\gamma}}^{\e,\delta}\le \frac{\delta}{\gamma}+\delta-\frac{1}{\gamma}.
		\end{align*}
		The first equality is due that $E_{\gamma}(\rho||\sigma)=\gamma E_{\frac{1}{\gamma}}(\sigma||\rho)$ \cite{hirche2023quantum}. At last, combing Theorem \ref{contractionco}, we have
		\begin{align*}
			\eta_{E_{\gamma}}^{\epsilon,\delta}\le \min((\gamma+1)\delta-\gamma,\frac{\delta}{\gamma}+\delta-1).
		\end{align*}
	\end{proof}
	
	At last, we consider and present the bounds of the sample complexity of symmetric and asymmetric hypothesis testing under $(\e,\d)$-QPrLDP for arbitrary couple of quantum states.
	\begin{Theorem}\label{bsfb}
		Assume $\epsilon>0,$ $p\in(0,1)$, $q=1-p$, and $\rho$ and $\sigma$ are states, then\begin{align*}
		\lceil\frac{2(e^{\e}-1+2E_{e^{\e}}(\r||\s))^2\ln\frac{\sqrt{pq}}{\alpha}}{(e^{\e}-1)^2E^2_{e^{\e}}(\rho||\sigma)}\rceil\ge	 SC_{(\rho,\sigma)}^{\mathcal{Q}^{\e,\d}}(\a,p,q)\ge f(\r,\s,\a,p,q).
		\end{align*}
		Here $f(\r,\s,\a,p,q)=\max\left(\frac{1-\frac{pq}{\a}}{-\ln(1-(2\delta-1)T(\rho,\sigma))},\frac{1-\frac{\a(1-\a)}{pq}}{2(1-\sqrt{1-(2\delta-1)T(\rho,\sigma)}},SC_{(\rho,\sigma)}(\alpha,p,q)\right).$
	\end{Theorem}
	\begin{proof}
		As the trace distance is monotone under a completely positive and trace-preserving, assume $SC_{{(\rho,\sigma)}}(\alpha,p,q)=z,$ then 
		\begin{align*}
			1-2z\le	||\otimes_{i=1}^n\mathcal{A}(p\rho^{\otimes n}-q\sigma^{\otimes n})||\le ||p\rho^{\otimes n}-q\sigma^{\otimes n}||,
		\end{align*}
		hence, $SC_{(\rho,\sigma)}(\alpha,p,q)$ is a lower bound of $SC_{(\rho,\sigma)}^{\epsilon,\delta}(\alpha,p,q)$. 
		
		Next due to $-\ln x\ge 1-x$ for $x>0,$
		\begin{align*}
			-\frac{1}{\ln \sqrt{F(\rho,\sigma)}}\le \frac{2}{d_B^2(\rho,\sigma)}\le \frac{2}{T^2(\r,\sigma)},
		\end{align*}
		where the last inequality is due to the Fuchs-van-de-Graaf inequalities, then we have an upper bound of $SC_{(\rho,\sigma)}(\alpha,p,q)$
		\begin{align*}
			SC_{(\rho,\sigma)}^{\mathcal{A}}\le \lceil\frac{2\ln\frac{\sqrt{pq}}{\alpha}}{T^2(\mathcal{A}(\rho),\mathcal{A}(\sigma))}\rceil.
		\end{align*}
		
		Let $\mathcal{A}=\mathcal{D}_p\circ\mathcal{M}$, and $\mathcal{M}(\cdot)$ is defined as follows,
		\begin{align*}
			\mathcal{M}(\cdot)=\mathrm{tr}(\Pi_{+}\cdot)\ket{+}\bra{+}+\mathrm{tr}(\Pi_{-}\cdot)\ket{-}\bra{-},
		\end{align*}
		here $\Pi_{+}$ denotes the projection operator $\Pi_{+}=\{\rho-e^{\e}\sigma> 0\}$, $\Pi_{-}=I-\Pi_{+}$, besides, $\ket{+}$ and $\ket{-}$ are pure states with $\ket{+}\perp\ket{-}$. Based on Corollary \ref{DpQPrLDP}, when $p=\frac{2E_{e^{\e}}(\r||\s)}{2E_{e^{\e}}(\r||\s)+e^{\e}-1}$, $\mathcal{A}$ is $(\epsilon,\delta)$-QPrLDP.
	
		\begin{align}
			&T^2(\mathcal{A}(\rho),\mathcal{A}(\sigma))\nonumber\\=&(1-p)^2\mathrm{tr}^2\Pi_{+}(\rho-\sigma)\nonumber\\
			\ge& (1-p)^2E^2_{e^{\e}}(\rho||\sigma),\label{f1}
		\end{align}
		For the upper bound,
		\begin{align*}
			SC_{(\rho,\sigma)}^{\mathcal{A}}\le& \inf_{\mathcal{A}\in \mathcal{Q}^{\e,\d}} \lceil\frac{2\ln\frac{\sqrt{pq}}{\alpha}}{T^2(\mathcal{A}(\rho),\mathcal{A}(\sigma))}\rceil\\
			\le&\lceil\frac{2\ln\frac{\sqrt{pq}}{\alpha}}{(1-p)^2E^2_{e^{\e}}(\rho||\sigma)}\rceil.\\
			=&\lceil\frac{2(e^{\e}-1+2E_{e^{\e}}(\r||\s))^2\ln\frac{\sqrt{pq}}{\alpha}}{(e^{\e}-1)^2E^2_{e^{\e}}(\rho||\sigma)}\rceil.
		\end{align*}The last inequality is due to (\ref{f1}).

		Hence, combining Lemma \ref{syqhtc}, we have
		\begin{align*}
			SC_{(\rho,\sigma)}^{\mathcal{A}}\ge& \frac{\ln(\frac{pq}{\alpha})}{-\ln F(\mathcal{A}(\rho),\mathcal{A}({\sigma}))},\\
			\ge&     \frac{\ln(\frac{pq}{\alpha})}{-\ln (1-T(\mathcal{A}(\rho),\mathcal{A}(\sigma)))}\\
			\ge&\frac{1-\frac{pq}{\a}}{-\ln(1-(2\delta-1)T(\rho,\sigma))},
		\end{align*}
here the first inequality is due to the Fuchs-van de Graaf inequalities \cite{fuchs1999}, the second inequality is due to Theorem \ref{contractionco}
		\begin{align*}
			SC_{(\rho,\sigma)}^{\mathcal{A}}\ge &	\frac{1-\frac{\a(1-\a)}{pq}}{d^2_B(\mathcal{A}(\rho),\mathcal{A}(\sigma))},\\
=&\frac{1-\frac{\a(1-\a)}{pq}}{2(1-\sqrt{F}(\mathcal{A}(\rho),\mathcal{A}(\sigma))}\\
\ge&\frac{1-\frac{\a(1-\a)}{pq}}{2(1-\sqrt{1-T(\mathcal{A}(\rho),\mathcal{A}(\sigma))}}\\
\ge&\frac{1-\frac{\a(1-\a)}{pq}}{2(1-\sqrt{1-(2\delta-1)T(\rho,\sigma)}},
		\end{align*}
		here the second inequality is due to the Fuchs-van de Graaf inequalities \cite{fuchs1999}, and the last inequality is due to Theorem \ref{contractionco}.
		
		that is, $$SC_{(\rho,\sigma)}^{\mathcal{A}}(\a,p,q)\ge \max\left(\frac{1-\frac{pq}{\a}}{-\ln(1-(2\delta-1)T(\rho,\sigma))},\frac{1-\frac{\a(1-\a)}{pq}}{2(1-\sqrt{1-(2\delta-1)T(\rho,\sigma)}},SC_{(\rho,\sigma)}(\alpha,p,q)\right).$$

		Hence, we finish the proof.
	\end{proof}
	
	With a similar method, we can generalize the above scenario to multiple quantum hypothesis testing.
	\begin{Corollary}
Assume $\e,\d\in [0,1]$, and $\{\rho_i\}_{i=1}^m$ is a tuple of quantum states with prior probabilities $p_i$, respectively, then its sample complexity of private multiple hypothesis testing  satisfies 
\begin{align*}
\max_{m\ne \tilde{m}}\frac{\ln\left(\frac{p_mp_{\tilde{m}}}{\e(p_m+p_{\tilde{m}})}\right)}{-\ln (1-(2\d-1)T(\rho_m||\rho_{\tilde{m}}))}	\le SC^{\e}_{((\rho_i)_{i=1}^m,(p_i)_{i=1}^m)}\le  \left\lceil\max_{m\ne \tilde{m}}\frac{2\ln\left(\frac{M(M-1)\sqrt{p_m}\sqrt{p_{\tilde{m}}}}{2\e}\right)}{-\ln (1-\frac{(e^{\e}-1)^2E^2_{e^{\e}}(\r||\s)}{(e^{\e}-1+2E_{e^{\e}}(\rho||\sigma))^2})^{\frac{1}{2}}}\right\rceil.
\end{align*}
	\end{Corollary}
	\begin{proof}
Based on Lemma \ref{symqhtc} and the proof method of Theorem \ref{bsfb}, we have
the lower bound of $SC_{((\mathcal{A}(\rho_i))_{i=1}^m,(p_i)_{i=1}^m)}^{\e}$,

\begin{align*}
	 SC_{((\mathcal{A}(\rho_i))_{i=1}^m,(p_i)_{i=1}^m)}^{\e}\ge& \inf_{\mathcal{A}\in \mathcal{Q}^{\e,\d}}\max_{m\ne \tilde{m}}\frac{\ln\left(\frac{p_mp_{\tilde{m}}}{\e(p_m+p_{\tilde{m}})}\right)}{-\ln F(\mathcal{A}(\rho_m),\mathcal{A}(\rho_{\hat{m}}))}\\
	 \ge&\inf_{\mathcal{A}\in \mathcal{Q}^{\e,\d}}\max_{m\ne \tilde{m}}\frac{\ln\left(\frac{p_mp_{\tilde{m}}}{\e(p_m+p_{\tilde{m}})}\right)}{-\ln (1-T(\mathcal{A}(\rho_m)||\mathcal{A}(\rho_{\tilde{m}})))}\\
	 \ge&\max_{m\ne \tilde{m}}\frac{\ln\left(\frac{p_mp_{\tilde{m}}}{\e(p_m+p_{\tilde{m}})}\right)}{-\ln (1-(2\d-1)T(\rho_m||\rho_{\tilde{m}}))}
\end{align*}
 
 Next we show the upper bound of $SC_{((\mathcal{A}(\rho_i))_{i=1}^m,(p_i)_{i=1}^m)}^{\e},$ Let $\mathcal{A}=\mathcal{D}_p\circ\mathcal{M}$, and $\mathcal{M}(\cdot)$ is defined as follows,
 \begin{align*}
 	\mathcal{M}(\cdot)=\mathrm{tr}(\Pi_{+}\cdot)\ket{+}\bra{+}+\mathrm{tr}(\Pi_{-}\cdot)\ket{-}\bra{-},
 \end{align*}
 here $\Pi_{+}$ denotes the projection operator $\Pi_{+}=\{\rho-e^{\e}\sigma> 0\}$, $\Pi_{-}=I-\Pi_{+}$, besides, $\ket{+}$ and $\ket{-}$ are pure states with $\ket{+}\perp\ket{-}$.  With a similar analysis, we have
 \begin{align*}
 	SC_{((\mathcal{A}(\rho_i))_{i=1}^m,(p_i)_{i=1}^m)}^{\e}\le& \inf_{\mathcal{A}\in \mathcal{Q}^{\e,\d}}\left\lceil\max_{m\ne \tilde{m}}\frac{2\ln\left(\frac{M(M-1)\sqrt{p_m}\sqrt{p_{\tilde{m}}}}{2\e}\right)}{-\ln F(\rho_m,\rho_{\tilde{m}}) }\right\rceil\\
 	\le &\left\lceil\max_{m\ne \tilde{m}}\frac{2\ln\left(\frac{M(M-1)\sqrt{p_m}\sqrt{p_{\tilde{m}}}}{2\e}\right)}{-\ln \sqrt{1-T^2(\rho||\sigma)} }\right\rceil\\
 	\le& \left\lceil\max_{m\ne \tilde{m}}\frac{2\ln\left(\frac{M(M-1)\sqrt{p_m}\sqrt{p_{\tilde{m}}}}{2\e}\right)}{-\ln (1-\frac{(e^{\e}-1)^2E^2_{e^{\e}}(\r||\s)}{(e^{\e}-1+2E_{e^{\e}}(\rho||\sigma))^2})^{\frac{1}{2}}}\right\rceil
 \end{align*}
 	\end{proof}
 	At last, we present the bounds of sample complexity of asymmetric hypothesis testing for arbitrary couple of quantum states.
 	\begin{Theorem}
 		Assume $\r$ and $\s$ are two states, then 
 		\begin{align*}
	ASC^{\e,\d}(\r,\s,\varphi,\vartheta)\ge	\max\{\frac{\log (2\d-1)T(\rho||\sigma)-(\a-1)\log(\a-1)+\a \log \a}{\a-1},ASC(\r,\s,\varphi,\vartheta)\},\\
		ASC^{\e,\d}(\r,\s,\varphi,\vartheta)\le \min\{\lceil\left(\frac{\ln\frac{\varphi^{}}{\vartheta}}{-\log(1-\frac{e^{\e}-1}{2E_{e^{\e}}(\r||\s)+e^{\e}-1}E_{e^{\e}}(\r||\s))}\right)\rceil,\lceil\left(\frac{\frac{\vartheta}{\varphi}}{-\log(1-\frac{e^{\e}-1}{2E_{e^{\e}}(\r||\s)+e^{\e}-1}E_{e^{\e}}(\s||\r))}\right)\rceil\}.
 		\end{align*}
 	\end{Theorem}
 	\begin{proof}
 		For the upper bound, we would apply the method in Theorem \ref{bsfb} and Lemma \ref{asyqhtc}. 	Let 
 		\begin{align*}
 			\mathcal{A}(\r)=&\mathcal{D}_p\circ\mathcal{M}(\r)\\
 			=&[(1-p)\mathrm{tr}(\Pi_{+}\rho)+\frac{p}{2}]\ket{+}\bra{+}+[(1-p)\mathrm{tr}(\Pi_{-}\rho)+\frac{p}{2}]\ket{-}\bra{-}\\
 \textit{here}\hspace{8mm}&			\mathcal{M}(\cdot)=\mathrm{tr}(\Pi_{+}\cdot)\ket{+}\bra{+}+\mathrm{tr}(\Pi_{-}\cdot)\ket{-}\bra{-},
 		\end{align*}
 		here $\Pi_{+}$ denotes the projection operator $\Pi_{+}=\{\rho-e^{\e}\sigma> 0\}$, $\Pi_{-}=I-\Pi_{+}$, besides, $\ket{+}$ and $\ket{-}$ are pure states with $\ket{+}\perp\ket{-}$. Based on Corollary \ref{DpQPrLDP}, when $p=\frac{2E_{e^{\e}}(\r||\s)}{2E_{e^{\e}}(\r||\s)+e^{\e}-1}$, $\mathcal{A}$ is $(\epsilon,\delta)$-QPrLDP. When $\a=\frac{1}{2}$,
 		\begin{align*}
 			D_{\a}(\mathcal{A}(\rho)||\mathcal{A}(\sigma))=&\frac{1}{\alpha-1}\log(x_1^{1-\a}x_2^{\a}+y_1^{1-\a}y_2^{\a}
 		)\\
 		\ge&-\frac{1}{2}\log(1-T^2(\mathcal{A}(\r)||\mathcal{A}(\s)))\\
 		\ge&-\frac{1}{2}\log(1-(1-p)^2E^2_{e^{\e}}(\r||\s)),
 		\end{align*}
 		here $x_1=	(1-p)\mathrm{tr}\Pi_{+}\rho+\frac{p}{2},x_2=(1-p)\mathrm{tr}^{}\Pi_{+}\s+\frac{p}{2},y_1=(1-p)\mathrm{tr}^{}\Pi_{-}\rho+\frac{p}{2},y_2=(1-p)\mathrm{tr}^{}\Pi_{-}\s+\frac{p}{2}$. In the second inequality, we apply the Fuchs-van de Graaf inequality \cite{fuchs1999} and $D_{\a}(\mathcal{A}(\r)||\mathcal{A}(\s))$ is nondecreasing when $\a\ge \frac{1}{2}$. Based on Lemma \ref{asyqhtc}, we have
 		\begin{align*}
 			ASC^{\e,\d}(\r,\s,\varphi,\vartheta)\le \min\{\lceil\left(\frac{2\ln\frac{\varphi^{}}{\vartheta}}{\log(1-(1-p)^2E^2_{e^{\e}}(\r||\s))}\right)\rceil,\lceil\left(\frac{2\ln\frac{\vartheta}{\varphi}}{-\log(1-(1-p)^2E^2_{e^{\e}}(\r||\s))}\right)\rceil\}.
 		\end{align*}
 		For the lower bound, with a similar method of Theorem \ref{bsfb}, we have
 		\begin{align*}
 				ASC^{\e,\d}(\r,\s,\varphi,\vartheta)\ge 	ASC(\r,\s,\varphi,\vartheta).
 		\end{align*}
 		Next as when $\alpha\ge 1,$
 		\begin{align*}
	\tilde{D}_{\alpha}(\mathcal{A}(\rho)||\mathcal{A}(\sigma))\ge& \frac{\log T(\mathcal{A}(\rho)||\mathcal{A}(\sigma))-(\a-1)\log(\a-1)+\a \log \a}{\a-1}\\
	\ge&\frac{\log (2\d-1)T(\rho||\sigma)-(\a-1)\log(\a-1)+\a \log \a}{\a-1}.
\end{align*}
 		The first inequality is due to that $\tilde{D}_{\alpha}$ is monotone in terms of $\alpha$ and $s$
 		Hence, we have
 		\begin{align*}
 			ASC^{\e,\d}(\r,\s,\e,\d)\ge	\frac{\log (2\d-1)T(\rho||\sigma)-(\a-1)\log(\a-1)+\a \log \a}{\a-1}.
 		\end{align*}
 	\end{proof}
 \section{Conclusion and Discussion}\label{VII}
 In this manuscript, we introduced and systematically investigated quantum probabilistic local differential privacy. First, we introduced the definition of quantum probabilistic local differential privacy. Then we derived several structural properties of the class of quantum superoperators with quantum probabilistic local differential privacy. We next characterized sufficient conditions when global and local depolarizing channels satisfy quantum probabilistic local differential privacy.
 We further investigated the quantum probabilistic local differential privacy. from information-theoretic perspectives. We subsequently established bounds on the sample complexity of symmetric quantum hypothesis testing under quantum probabilistic local differential privacy. In addition, we derived a lower bound on the probabilistic privacy contraction coefficient in terms of the hockey-stick divergence. These results provide a theoretical foundation for understanding the privacy–distinguishability tradeoff in quantum information processing and may facilitate the development of privacy-preserving quantum superoperators and quantum machine-learning protocols.
 
 The other characterization of quantum probabilistic local differential privacy can be defined as follows, for any couple of quantum states $\rho$ and $\sigma$,  and $\{M_k\}_{k\in \mathcal{O}}$ is an arbitrary POVM with finite outcomes,
 \begin{align}
 	\hat{D}_{\e,\delta}(\mathcal{A}(\r)||\mathcal{A}(\s))=\sum_{k\in T}&	\mathrm{ tr}[M_k\mathcal{A}(\rho)]\le \delta,\hspace{3mm}.\label{f2}
 \end{align}
 Here $T=\{k|\mathrm{tr}M_k\mathcal{A}(\rho)\ge e^{\epsilon}\mathrm{tr}M_k\mathcal{A}(\sigma)\}$. Next $\hat{D}_{\epsilon,\delta}(\mathcal{A}(\rho)||\mathcal{A}(\delta))$ can be characterized as follows,
 \begin{align*}
 		\hat{D}_{\e,\delta}(\mathcal{A}(\rho)||\mathcal{A}(\sigma))=\max_{0\le M\le I}\{\mathrm{tr}M\mathcal{A}(\rho)|\mathrm{ tr}M(\mathcal{A}(\rho)-e^{\e}\mathcal{A}(\sigma))\ge 0\}
 \end{align*}
  its dual formulation can be expressed as 
 \begin{align*}
 	\hat{D}_{\e,\delta}(\mathcal{A}(\rho)||\mathcal{A}(\sigma))=\inf_{\lambda\ge 0}[(1+\lambda)\mathcal{A}(\rho)-\lambda e^{\e}\mathcal{A}(\sigma)]_{+}.
 \end{align*}

 An interesting problem is to study other properties of the quantum probabilistic differential privacy, such as bounds of the sample complexity and the probabilistically private contraction coefficients under some a distance. 
	\section{APPENDIX }
	\begin{Lemma}\label{cnwb}
		Assume $\rho$ is a state, $B$ and $C$ are two Hermitian matrices, then
		\begin{align*}
\mathrm{tr}\rho\{B+C\ge 0\}\le			\mathrm{tr}\rho\{B\ge 0\}+\mathrm{tr}\rho\{C\ge 0\}
		\end{align*}
	\end{Lemma}
	\begin{proof}
		\begin{align*}
		\mathrm{tr}\rho\{B+C\ge 0\}\le &\lambda_{max}(\r)\rank((B+C)_{+})\\
			\le&\lambda_{max}(\r)(rank(B_{+})+rank(C_{+}))\\
			\le& \frac{\lambda_{max}(\r)}{\lambda_{min}(\r)}\mathrm{tr}\r[\{B\ge 0\}+\{C\ge 0\}],
		\end{align*}
		here the first inequality is due to that $\rho\le \lambda_{max}(\r)I$, the second inequality is due to that $\rank((B+C)_{+})\le rank(B_{+})+rank(C_{+}),$ the last inequality is due to that $\rho\ge \lambda_{min}(\rho)I.$
	\end{proof}
	\begin{Lemma}\label{dtf}
	Assume $\rho,\sigma\in \mathcal{D}(\mathcal{H})$, $\frac{1}{2}||\rho-\sigma||_1\le\eta$, then 
		\begin{align*}
			\mathrm{ tr}\rho\sigma\le 1-\eta^2.
		\end{align*}
	\end{Lemma}
	\begin{proof}
		\begin{align*}
			\mathrm{tr}\rho\sigma\le F^2(\rho,\sigma)\le 1-\eta^2,
		\end{align*}
	here	the second inequality is due to the Fuchs-van de Graaf inequality \cite{fuchs1999}.
	\end{proof}
	\begin{Lemma}\label{lnain}
		Assume $\rho$ and $\sigma$  are two states with $\frac{1}{2}||\rho-\sigma||_1\le \chi$, $\gamma\ge 1,$ then
		\begin{align*}
				\inf_{s\in [0,1]}\gamma^{2-s}\mathrm{tr}\rho^s\sigma^{1-s}\le\gamma(1+\chi)
		\end{align*}
	\end{Lemma}
	\begin{proof}
			Assume $\chi\triangle=(\rho-\sigma)_{+}$ is the positive part of $\rho-\sigma,$ 
		\begin{align*}
			\inf_{s\in [0,1]}\gamma^{2-s}\mathrm{tr}\rho^s\sigma^{1-s}\le&\inf_{s\in [0,1]}\gamma^{2-s}\mathrm{tr}\rho^s(\rho+\chi\triangle)^{1-s}\\
			\le&\inf_{s\in [0,1]}\gamma^{2-s}\mathrm{tr}(\rho+\chi\triangle)\\
			\le& \gamma(1+\chi),
		\end{align*}
		the first inequality is due to that $\frac{1}{2}||\rho-\sigma||_1\le \chi,$ the second inequality is due to that the function $x^s(s\in (0,1))$ is operator monotone, the last inequality is due to that $\mathrm{ tr}\chi\triangle\le\chi$.
	\end{proof}
	\begin{Lemma}\label{na}
		Assume $A$ and $B$ are two positive operators, then for all $s\in [0,1]$, then 
		\begin{align*}
			\mathrm{tr}B\{A-B\ge 0\}\le\inf_{s\in [0,1]}\mathrm{tr}A^sB^{1-s}.
		\end{align*}
	where the infimum takes over all the values in $[0,1].$
	\end{Lemma}
	\begin{proof}
		Let $P=\{A-B\ge 0\}$, then 
			\begin{align*}
			\mathrm{tr}A^sB^{1-s}\ge& \mathrm{tr}\frac{A+B-|A-B|}{2}\\
			=&\mathrm{ tr}\frac{AP+A(I-P)+B-(A-B)P}{2}\\
			=&\mathrm{ tr}\frac{(A+B)(I-P)+2BP}{2}\ge \mathrm{tr}BP.
		\end{align*}
		The first inequality is based on \cite{audenaert2007}. As the first inequality is valid for all $s\in [0,1]$, we finish the proof.
	\end{proof}
	\begin{Lemma}\label{de}
		Assume $\rho$ and $\sigma$ are two states, and $\gamma\ge 1$, then
		\begin{align*}
		D_{\gamma}(\rho||\sigma)\le E_{\gamma}(\rho||\sigma)+\inf_{s\in [0,1]}\gamma^{2-s}\mathrm{ tr}\rho^s\sigma^{1-s}.
		\end{align*}
		If $\frac{1}{2}||\rho-\sigma||_1\le \epsilon,$ then
		\begin{align*}
			D_{\gamma}\le\epsilon+\gamma(1+\epsilon)
		\end{align*}
	\end{Lemma}
	\begin{proof}
		\begin{align*}
			D_{\gamma}(\rho||\sigma)=&\mathrm{ tr}\rho\{\rho-\gamma\sigma\ge 0\}\\
			=&\mathrm{ tr}(\rho-\gamma\sigma)\{\rho-\gamma\sigma\ge 0\}+\gamma\mathrm{tr}\sigma\{\rho-\gamma\sigma\ge 0\}\\
			\le &E_{\gamma}(\rho||\sigma)+\inf_{s\in [0,1]}\gamma^{2-s}\mathrm{ tr}\rho^s\sigma^{1-s}.
		\end{align*}
		The last inequality is due to Lemma \ref{na}.

		If $\frac{1}{2}||\rho-\sigma||_1\le \epsilon$, 
		\begin{align*}
				D_{\gamma}\le& E_{\gamma}(\rho||\sigma)+\gamma(1+\epsilon)\\
			\le&\epsilon+\gamma(1+\epsilon),
		\end{align*}
		the first inequality is due to Lemma \ref{lnain}, the lastinequality is due to that $E_{\gamma}(\rho||\sigma)\le\frac{1}{2}||\rho-\sigma||_1\le \epsilon$.
	\end{proof}
	
	\begin{Lemma}\label{l2}
		Assume both $M_1$ and $M_2$ are Hermitian, then 
	 $\supp(M_1)_{+}\otimes\supp(M_2)_{+}=\supp(({M_1})_{+}\otimes ({M_2})_{+}),$ that is, $\{M_1\otimes M_1\ge 0\}=\{M_1\ge 0\}\otimes\{M_2\ge 0\}.$
	\end{Lemma}
	\begin{proof}	
	Assume ${M_1}_{+}=\sum_i \lambda_i\ket{i}\bra{i}$ and ${M_2}_{+}=\sum_k\mu_k\ket{k^{'}}\bra{k^{'}}$, then
		\begin{align*}
			&\supp(M_1)_{+}=span\{\ket{i}\},\\ &\supp(M_2)_{+}=span\{\ket{k^{'}}\},\\
		&	\supp(M_1)_{+}\otimes\supp(M_2)_{+}\\
			=&span\{\ket{i}\}\otimes span\{\ket{k^{'}}\},\\
		=&\supp(({M_1})_{+}\otimes ({M_2})_{+})\\
		=&\supp(M_1\otimes M_2)_{+}.
		\end{align*}
Hence, we finish the proof.
	\end{proof}
	\begin{Lemma}
		Assume $A$ and $B$ are semidefinite positive with $[A,B]=0$, then
		\begin{align*}
			\{A\ge 0\}+\{B\ge 0\}=\{A+B\ge 0\}.
		\end{align*}
		However, when $[A,B]\ne 0$, the above equality maybe not valid.
	\end{Lemma}
	
	\begin{proof}
	As $[A,B]=0$, we can always assume $A=\sum_i \mu_i \ket{i}\bra{i}$ and $B=\sum_k\phi_k\ket{k}\bra{k}$,
	\begin{align*}
		\{A>0\}+\{B>0\}-\{A+B>0\}=&\sum_{\{i|\mu_i>0\}}\ket{i}\bra{i}+\sum_{\{i|\phi_i>0\}}\ket{i}\bra{i}-\sum_{\{i|\mu_i+\phi_i>0\}}\ket{i}\bra{i}\\
		\ge& 0.
	\end{align*}
	If $k\in \{i|\mu_i+\phi_i>0\},$ $\mu_k>0$ or $\phi_k>0$. Hence, the inequality is valid.
	\end{proof}
	\begin{Lemma}\label{lds}
		Assume $\gamma>0,$ $\rho$, $\sigma$ and $\phi$ are quantum states, then 
		\begin{itemize}
			\item[(i)]
			\begin{align*}
				D_{\gamma}(\rho||\sigma)+D_{\frac{1}{\gamma}}(\sigma||\rho)\ge \gamma+E_{\gamma}(\rho||\sigma),
			\end{align*}
			\item[(ii)] if $\rho=\sum_x p_x\rho_x$ and $\sigma=\sum_x p_x\sigma_x$, 
				\begin{align*}
				D_{\gamma}(\rho||\sigma)\le \sum_xp_x\mathrm{tr}(\rho_x-\gamma\sigma_x)_{+}+\inf_{s\in [0,1]}\gamma^{1-s}\mathrm{ tr}\rho^s\sigma^{1-s}.
			\end{align*}
			\item[(iii)] \begin{align*}
				D_{\gamma}(\rho\otimes\phi||\sigma\otimes\phi)=D_{\gamma}(\rho||\sigma)
			\end{align*}
		\end{itemize}
	\end{Lemma}
	\begin{proof}
		\begin{itemize}
			\item[(i).] 
			\begin{align*}
				&D_{\gamma}(\rho||\sigma)\\
				=&\mathrm{tr}\rho\{\rho-\gamma\sigma\ge 0\}\\
				=&\mathrm{tr}(\rho-\gamma\sigma)\{\rho-\gamma\sigma\ge 0\}+\gamma\mathrm{tr}\sigma(I-\{\sigma-\frac{1}{\gamma}\rho> 0\})\\
				\ge& E_{\gamma}(\rho||\sigma)+\gamma-D_{\frac{1}{\gamma}}(\sigma||\rho),
			\end{align*}
			
			\item[(ii)]\begin{align*}
				D_{\gamma}(\rho||\sigma)=&\mathrm{ tr}\rho\{\rho-\gamma\sigma\ge0\}\\
				=&\mathrm{ tr}[\sum_x p_x(\rho_x-\gamma\sigma_x)\{\sum_x p_x(\rho_x-\gamma\sigma_x)\ge 0\}]\\+&\mathrm{ tr}[\sum_x p_x\gamma\sigma_x\{\sum_x p_x(\rho_x-\gamma\sigma_x)\ge 0\}]\\
			\le&\sum_xp_x\mathrm{tr}(\rho_x-\gamma\sigma_x)_{+}+\inf_{s\in [0,1]}\gamma^{1-s}\mathrm{ tr}\rho^s\sigma^{1-s}.
			\end{align*}
			here the first inequality is due to Lemma \ref{l2}, the last equality is due to the definiton of $D_{\gamma}(\cdot||\cdot)$.
			\item[(iii).]
			\begin{align*}
				D_{\gamma}(\rho\otimes\phi||\sigma\otimes\phi)=&\mathrm{tr}[\rho\otimes\phi\{\rho\otimes\phi-\gamma\sigma\otimes\phi\}]\\
				=&\mathrm{tr}[\rho\otimes\phi\{(\rho-\gamma\sigma)\otimes\phi\ge 0\}]
\\=&\mathrm{tr}\rho\{\rho-\gamma\sigma\}=D_{\gamma}(\rho||\sigma).
			\end{align*}
		\end{itemize}
	\end{proof}
	 
	\begin{Lemma}\label{lmh}
		Assume $\rho$ and $\sigma$ are two states, $\m\ge 1$, then for any $n\in \mathbb{N}$,
		\begin{align*}
			E_m(\rho^{\otimes n}||\sigma^{\otimes n})\ge E^n_{m}(\rho||\sigma).
		\end{align*}
	\end{Lemma}
	\begin{proof}
		\begin{align*}
			E_m(\rho^{\otimes n}||\sigma^{\otimes n})=&\mathrm{tr}(\rho^{\otimes n}-m\sigma^{\otimes n})_{+}\\
			=&\max_{0\le \Lambda\le I}\mathrm{tr}M(\rho^{\otimes n}-m\sigma^{\otimes n})\\
			\ge&\mathrm{tr}N^{\otimes n}({\rho}^{\otimes n}-m\sigma^{\otimes n})\\
			=&\mathrm{ tr}^nN\rho-m\mathrm{ tr}^nN\sigma\\
			=&\mathrm{ tr}^n(N\rho-mN\sigma+mN\sigma)-m\mathrm{ tr}^nN\sigma\\
			\ge&\mathrm{ tr}^n(N\rho-mN\sigma)+m^n\mathrm{ tr}^nN\sigma-m\mathrm{ tr}^nN\sigma\\
			\ge&\mathrm{tr}^nN(\rho-m\sigma)=E_m^n(\rho||\sigma).
		\end{align*}
		In the first inequality, $N$ is the optimal such that $E_m(\rho||\sigma)=\mathrm{ tr}N(\rho-m\sigma)$, the last inequality is due to that $m\ge 1.$
	\end{proof}

	\begin{Lemma}\label{sandwichedgamma}
		Assume $\rho$ and $\sigma$ are two states, and $\alpha,\gamma>1$, then
		\begin{align*}
	E_{\gamma}(\rho^{\otimes n}||\sigma^{\otimes n})\le \frac{(\a-1)^{\a-1}}{\a^{\a}}e^{n(\a-1)\tilde{D}_{\a}(\r||\s)}\gamma^{1-\a}.
		\end{align*}
		When $n=1,$
		\begin{align*}
			E_{\gamma}(\rho||\sigma)\le \frac{(\a-1)^{\a-1}}{\a^{\a}}e^{(\a-1)\tilde{D}_{\alpha}(\rho||\s)}\gamma^{1-\a}.
		\end{align*}
		Besides, when $\a\in (1,2],$
		\begin{align*}
			E_{\gamma}(\rho^{\otimes n}||\sigma^{\otimes n})\le \frac{(\a-1)^{\a-1}}{\a^{\a}}e^{n(\a-1){D}_{\a}(\r||\s)}\gamma^{1-\a}.
		\end{align*}
	\end{Lemma}
	\begin{proof}
		Let $M=\{\rho-\gamma\sigma\ge 0\}$, and $p=\mathrm{tr}M\rho,$ $q=\mathrm{tr}M\sigma$, then
		\begin{align*}
			\tilde{D}_{\alpha}(\rho||\sigma)=&\frac{1}{\alpha-1}\log\mathrm{tr}(\sigma^{\frac{1-\alpha}{2\alpha}}\rho\sigma^{\frac{1-\alpha}{2\alpha}})^{\alpha}\\
			\ge &\frac{1}{\alpha-1}\log[p^{\alpha}q^{1-\alpha}+(1-p)^{\alpha}(1-q)^{1-\alpha}]\ge \frac{1}{\alpha-1}\log p^{\alpha}q^{1-\alpha},\\
			\Longrightarrow&\tilde{D}_{\alpha}(\rho||\sigma)\ge \frac{\a}{\alpha-1}\log p-\log q,\\
			\Longleftrightarrow& p\le e^{\frac{\alpha-1}{\alpha}\tilde{D}_{\alpha}(\rho||\sigma)}q^{\frac{\alpha-1}{\alpha}}.
		\end{align*}
		Then
		\begin{align*}
			E_{\gamma}(\rho||\sigma)=p-\gamma q\le e^{\frac{\alpha-1}{\alpha}\tilde{D}_{\alpha}(\rho||\sigma)}q^{\frac{\alpha-1}{\alpha}}-\gamma q,
		\end{align*}
		let $h(q)=e^{\frac{\alpha-1}{\alpha}\tilde{D}_{\alpha}(\rho||\sigma)}q^{\frac{\alpha-1}{\alpha}}-\gamma q$, 
		\begin{align*}
			h^{'}(q)=& \frac{\alpha-1}{\alpha}e^{\frac{\alpha-1}{\alpha}\tilde{D}_{\alpha}(\rho||\sigma)}q^{-\frac{1}{\alpha}}-\gamma,\\
			h^{''}(q)=&  -\frac{\alpha-1}{\alpha^2}e^{\frac{\alpha-1}{\alpha}\tilde{D}_{\alpha}(\rho||\sigma)}q^{-\frac{1}{\alpha}-1}\le0,
		\end{align*}
		From the above analysis, $h(q)$ is attains the maximum when $h^{'}(q)=0,$ 
		\begin{align*}
			\max_qh(q)=\frac{(\a-1)^{\a-1}}{\a^{\a}}e^{(\a-1)\tilde{D}_{\a}(\r||\s)}\gamma^{1-\a}.
		\end{align*}
		Hence,
		\begin{align}
			E_{\gamma}(\rho||\sigma)\le \frac{(\a-1)^{\a-1}}{\a^{\a}}e^{(\a-1)\tilde{D}_{\a}(\r||\s)}\gamma^{1-\a}. \label{af1}
		\end{align}
		Based on the above analysis and the additivity of $\tilde{D}_{\a}(\cdot||\cdot)$ when $\alpha>1,$
		\begin{align*}
			E_{\gamma}(\rho^{\otimes n}||\sigma^{\otimes n})\le \frac{(\a-1)^{\a-1}}{\a^{\a}}e^{n(\a-1)\tilde{D}_{\a}(\r||\s)}\gamma^{1-\a}.
		\end{align*}
		
		As the Petz-Renyi relative entropy $D_{\alpha}(\rho||\sigma)=\frac{1}{\alpha-1}\log\mathrm{tr}\rho^{\alpha}\sigma^{1-\alpha}$ satisfies data processing inequality when $\alpha\in (1,2],$ by a simple method above, we can finish the proof.
	\end{proof}

	\bibliographystyle{IEEEtran}
	\bibliography{ref}
	\end{document}